\documentclass[a4paper,12pt]{article}
\usepackage{graphicx}
\usepackage{mathtools}
\usepackage{authblk}
\usepackage{braket}

\usepackage[utf8]{inputenc}
\usepackage{amsmath}
\usepackage{ulem}
\usepackage{bm}      % for \bm macro
\usepackage[colorlinks,linkcolor=blue,citecolor=red]{hyperref}
\usepackage{color}
\usepackage{graphicx}
\usepackage{amsfonts}
\usepackage{soul}
\usepackage{comment}
\usepackage{amssymb}
\usepackage{xcolor}
\usepackage{ifthen}

\usepackage[paperwidth=199.8mm,paperheight=287mm,centering,hmargin=22.5mm,vmargin=2cm]{geometry}
\usepackage{theorem}
\newcommand{\supp}[0]{\text{supp}}

\newtheorem{definition}{Definition}[section]

\newtheorem{lemma}[definition]{Lemma}
\newtheorem{algorithm}[definition]{Algorithm}

\newtheorem{theorem}[definition]{Theorem}
\newtheorem{corollary}[definition]{Corollary}

\newcommand{\smallblacksquare}{\hfill\rule{2.2mm}{2.2mm}}
\newenvironment{proof}[1][]
  {\par\noindent\textbf{Proof\ifthenelse{\equal{#1}{}}{}{ of #1}:}\quad}
  {\smallblacksquare\par}
\newcommand{\gammat}{\gamma_{\scriptscriptstyle \mathcal{T}}}

\title{High-Temperature Fermionic Gibbs States are Mixtures of Gaussian States}

\author[1,2]{Akshar Ramkumar}
\author[1,3]{Yiyi Cai}
\author[5,6,7]{Yu Tong}
\author[1,4]{Jiaqing Jiang \thanks{jiaqingjiang95@gmail.com}}

\affil[1]{Institute for Quantum Information and Matter, California Institute of Technology}
\affil[2]{Division of Physics, Mathematics, and Astronomy, California Institute of Technology}
\affil[3]{Department of Electrical Engineering, California Institute of Technology}
\affil[4]{Department of Computing and Mathematical Sciences, California Institute of Technology}
\affil[5]{Department of Mathematics, Duke University}
\affil[6]{Department of Electrical and Computer Engineering, Duke University}
\affil[7]{Duke Quantum Center, Duke University}
\date{\vspace{-5ex}}

\begin{document}

\maketitle
\begin{abstract}
    Efficient simulation of a quantum system generally relies on structural properties of the quantum state. Motivated by the recent results by Bakshi et al. on the sudden death of entanglement in high-temperature Gibbs states of quantum spin systems, we study the high-temperature Gibbs states of bounded-degree local fermionic Hamiltonians, which include the special case of geometrically local fermionic systems.  We prove that at a sufficiently high temperature that is independent of the system size, the Gibbs state is a probabilistic mixture of fermionic Gaussian states. This forms the basis of an efficient classical algorithm to prepare the Gibbs state by sampling from a distribution of fermionic Gaussian states. 
    As a contrasting example, we show that high-temperature Gibbs states of the Sachdev-Ye-Kitaev (SYK) model are not convex mixtures of Gaussian states. 
\end{abstract}

\newpage
\tableofcontents
\newpage

% \red{To do
% \begin{itemize}
%     \item Add subtitles for sec 5.
%     \item (By the end of the procedure, classical sampling), alternative method for sampling
%     \item Weak interaction case: Discuss with Jeongwon Haah?
%     \item Unify the usage of Eq. vs Equation
% \end{itemize}

% }

% \YT{should number more equations}

\section{Introduction}

% Structural results for quantum states are crucial to understanding the difficulties of preparing such states, examples include entanglement area law for 1D gapped hams, which lead to efficient RG, CMI decay leads to efficient prep of thermal states, and recent high temp Gibbs state results
Studying quantum many-body systems through computational means has been a central topic in computational physics, chemistry, and quantum computing. Mathematically rigorous results that guarantee the efficient computation of ground state and Gibbs state properties typically rely on structural knowledge of these states. 
For example, the ground state of a 1D gapped local Hamiltonian can be efficiently obtained, as the area-law entanglement scaling ensures that it admits a matrix product state representation \cite{landau2015polynomial,hastings2007area}. Moreover, Gibbs states can be efficiently prepared on quantum computers when they exhibit exponential decay of correlations and satisfy  certain Markov property \cite{BrandaoKastoryano2019finite,chen2025quantum}.
% \jiaqing{I rephrase the sentences to make it more consistent with the sentence ``rigorous results that guarantee \jnew{the} efficient computation of ground state and Gibbs state properties typically rely on structural knowledge of these states. "}
% \jold{For example, the ground state of a 1D gapped local Hamiltonian can be efficiently obtained in the form of a matrix-product state because it satisfies the area-law entanglement scaling \cite{landau2015polynomial}, and Gibbs states with exponential decay of correlation and the uniform Markov property can be efficiently prepared on quantum computers \cite{BrandaoKastoryano2019finite}.} 

% We focus on fermions rather than spins. Why fermions are important (chemistry etc.). Classical heuristics, lack of efficient algorithm, recent efficient quantum algorithm for the weakly interacting case, provably efficient classical algorithm for weakly interacting fermions? First step, structural results for the Gibbs state. Focus on the high temperature case. 

Recently, for spin systems, a structural result demonstrated that high temperature Gibbs states are  unentangled and can be  efficiently prepared by classical computers~\cite{bakshi2024high}. More specifically, it was shown that for geometrically local Hamiltonians at a sufficiently high temperature that is independent of the system size, the Gibbs state is a classical distribution over product states and thus exhibits zero entanglement.  By efficiently sampling from this distribution, Bakshi et al.~\cite{bakshi2024high} extend this structural result into a classical algorithm for sampling and estimating observable values of the Gibbs state. While prior works~\cite{RouzeFrancaAlhambra2024efficient, RouzeFrancaAlhmabra2024optimal,kastoryano2016quantum} have shown that high-temperature Gibbs states can be efficiently prepared on quantum computers, the structural result and efficient classical algorithms in~\cite{bakshi2024high} suggest that the potential quantum advantage for quantum Gibbs sampling~\cite{chen2023quantum,jiang2024quantum,ding2025efficient} should be investigated in other parameter regimes.
In this work, we investigate analogues of the structural results and classical algorithms for high-temperature Gibbs states for fermionic systems. Fermionic Hamiltonians encompass a broad class of quantum systems of practical interest, including models for interacting electrons such as the electronic structure Hamiltonian~\cite{lin2019mathematical,whitfield2013computational,o2022intractability} and the Fermi-Hubbard model~\cite{qin2022hubbard,stanisic2022observing}.  In contrast to spin systems, where operators acting on different sites commute, fermionic systems are described by Majorana operators, which anti-commute and become non-local when mapped to the qubit setting: using the Jordan–Wigner transformation~\cite{nielsen2005fermionic}, the $2n$ Majorana operators  $\gamma_1,...,\gamma_{2n}$ can be mapped to tensor products of Pauli operators on $n$-qubit system, with $\gamma_{2j-1}=Z_1...Z_{j-1}X_j, \gamma_{2j}=Z_1...Z_{j-1}Y_j$. 
%\jiaqing{is it ok if we write this equation in the introduction? I am trying to make it understandable for people who don't know fermions}
%Local fermionic Hamiltonian is a sum of terms which consist of constant number of Majorona operators and
This anti-commutation relation and the non-local nature of Majorana operators induce distinct structural and entanglement properties for fermionic systems, raising the natural question of whether the results obtained for high-temperature Gibbs states in spin systems~\cite{bakshi2024high} still hold for fermionic systems. Specifically, are high-temperature fermionic Gibbs states likewise unentangled? Is there an analogous structural result for fermionic systems? If so, can such fermionic Gibbs states also be efficiently prepared using classical algorithms, thereby excluding the potential quantum advantage for these systems?

\vspace{0.5em}

% % Main result

\vspace{0.3em}

\subsection{Main results}
We show that at sufficiently high temperatures independent of the system size, the Gibbs state of a geometrically local fermionic Hamiltonian is a classical distribution over a special class of fermionic Gaussian states and can be efficiently sampled.
Gaussian states are the ground states and Gibbs states of 2-local fermionic Hamiltonians (free fermions); they are structurally simple, efficiently representable on classical computers, and are often considered the fermionic analogue of product states~\cite{surace2022fermionic}.

More specifically, we prove the following: 
\begin{theorem}[Informal version of Theorem \ref{thm:structural}]\label{thm:intro_separable}
     Consider a system of $2n$ Majorana operators on a constant-dimension lattice, and a fermionic Hamiltonian $H=\sum_a H_a$ where each $H_a$  involves only a constant number of Majorana operators on nearby lattice sites. Then, there exists a constant critical temperature such that, above this temperature,  the Gibbs state can be written as 
     \begin{align}
         \rho_\beta = \frac{e^{-\beta H}}{\text{tr}(e^{-\beta H})} =\sum_i q_i \rho^{(i)}, 
     \end{align}
     where $\{q^{(i)}\}_i$ is a classical distribution, $\rho^{(i)}$ belongs to a special class of Gaussian states, which can be defined w.r.t a partial matching $M^{(i)}$ of Majorana operators,  
     \begin{align}
         \rho^{(i)} = \frac{1}{2^n} \prod_{(k,l)\in M^{(i)}} \left(I+i\gamma_k \gamma_l\right) \label{eq:Gaussian_State}
     \end{align}
\end{theorem}

%\jiaqing{is it too technical? The main reason I write the above theorem is  I want the following comparison.  }

%Note that the partial matching $M^{(i)}$ may be different for different $i$.
It is worth noting that in the above theorem, the partial matching \( M^{(i)} \) can vary arbitrarily with \( i \). This leads to the key conceptual and technical distinction between the fermionic and spin cases. In the spin setting~\cite{bakshi2024high}, high-temperature Gibbs states are classical distributions over product states, all defined with respect to a fixed tensor product structure (e.g., qubit $1 \otimes$ qubit $2 \otimes \cdots \otimes$ qubit $n$). In contrast, in the fermionic case, each Gaussian state \( \rho^{(i)} \) corresponds to a different matching \( M^{(i)} \) of Majorana operators. While it is possible to reorder the Majorana operators so that a specific $\rho^{(i)}$ becomes a product state under the Jordan--Wigner transformation, there is no obvious way that transforms all \( \rho^{(i)} \)  into product states simultaneously.

Based on the structural result, we also show an efficient classical algorithm for sampling from the distribution of Gaussian states:
\begin{theorem}[Informal version of Theorem \ref{thm:efficient_gibbs_sampling}]\label{thm:intro_algo}
    For the same class of fermionic Hamiltonians $H = \sum_{a} H_a$, there exists an efficient classical algorithm to sample from the distribution $q_i$ in Theorem \ref{thm:intro_separable} up to $\epsilon$ precision, and thus the Gibbs states $\rho_\beta$ can be efficiently sampled on classical computers.
\end{theorem}

In the context of the results from \cite{bakshi2024high} for high-temperature Gibbs state in quantum spin systems, one might try to transform the fermionic Hamiltonian to a spin Hamiltonian and then extend their results to the fermionic setting directly. This approach encounters different difficulties based on the different fermionic encodings used. If the Jordan-Wigner transformation is used, then the locality of the Hamiltonian will no longer be preserved except for 1D quantum systems, which violates the locality requirement in \cite{bakshi2024high}. If the locality-preserving encoding of \cite{VerstraeteCirac2005mapping} is used, then ancilla spins are needed and these spins need to be in the ground state of an additional Hamiltonian, which violates the high-temperature requirement of \cite{bakshi2024high}. Therefore the results in this work do not follow from such a simple argument.

One may also ask whether the high-temperature fermionic Gibbs states can be represented as a mixture of unentangled states, as done in \cite{bakshi2024high}. Here we present an informal argument, which is independent of the rest of the paper, to show that the fermionic superselection rule precludes such a representation. The fermionic superselection rule forbids superposition between even and odd parity sectors. Mathematically, for a system consisting of modes $\gamma_1,\gamma_2,\cdots,\gamma_{2n}$, this means any valid quantum state has to commute with the parity operator $i^n \gamma_1\gamma_2\cdots\gamma_{2n}$. If a fermionic system is decomposed into subsystems $S_1,S_2,\cdots,S_{\ell}$ with parity operators $P_1,P_2,\cdots,P_{\ell}$ respectively, and a quantum state is unentangled across these subsystems, then it needs to simultaneously commute with parity operators $P_j$. If a Gibbs state can be expressed as a mixture of such states, then it also needs to commute with all operators $P_j$ by linearity. Since the Gibbs state is proportional to $e^{-\beta H}$, we also need $[H,P_j]=0$. This requirement is obviously not satisfied by all Hamiltonians, since a single quadratic term across two subsystems $j,k$ already fails to commute with both $P_j$ and $P_k$.

\vspace{0.5em}

%\YC
In addition to our structural and sampling results for bounded-degree local fermionic Hamiltonians, we present a complementary negative result for the Sachdev-Ye-Kitaev (SYK) model. As a contrasting example, we find that high-temperature Gibbs states of the Sachdev-Ye-Kitaev (SYK) model  are \textit{not} mixtures of Gaussian states, as stated in Theorem \ref{thm:intro_SYK}.
Note that while each interaction term in SYK is $q$-local, the model has unbounded degree thus Theorem \ref{thm:intro_separable} and Theorem \ref{thm:intro_algo} do not apply.
This shows that the structural phenomena we establish for bounded-degree local fermionic Hamiltonians break down in strongly interacting systems, such as SYK. The proof builds on prior work \cite{hastings2022optimizing,herasymenko2023optimizing,ding2025optimizing} that studies the ground state of the SYK model.

%We show that, unlike the bounded-degree case investigated above, the Gibbs states of SYK at high temperature cannot be expressed as mixtures of Gaussians.

\begin{theorem}[Informal version of Theorem \ref{thm:SYK_gaussian}]\label{thm:intro_SYK}
Let $q \geq 12$ be an even integer, and let $SYK_q$ denote the distribution over $q$-local Hamiltonians with unbounded degree on $2n$ Majorana modes defined by 
\begin{equation}
    H =  i^{q/2} \sum_{1\leq k_1<...<k_q\leq 2n} J_{k_1 \dots k_q}\gamma_{k_1}\dots \gamma_{k_q}. 
\end{equation}
where each coefficient $J_{k_1 \dots k_q}$ is an independent standard Gaussian variable with mean $0$ and variance $1$.
For $\epsilon>0$ and any inverse temperature $\beta = \Omega(n^{-3q/4 + 1 + \epsilon })$, with high probability over $H \sim \text{SYK}_q$ 
%there is high probability over $H \sim \text{SYK}_q$ that
there exists a $0 < \beta' < \beta$ for which the Gibbs state
\begin{equation}
\rho_{\beta'} = \frac{e^{-\beta'  H}}{\text{tr}\left(e^{-\beta' H}\right)}
\end{equation}
is not a mixture of Gaussian states.
\end{theorem}
Note that
the SYK model in physics and computer science literature is typically normalized by a coefficient either scaling as $n^{(q-1)/2}$ or $n^{q/2}$. 
With either choice of normalization, the above theorem implies as a corollary that for any \underline{constant} temperature $\beta^{-1}$, there exists a $0<\beta'< \beta$ for which $\rho_{\beta'}$ is not a mixture of Gaussian states.

\vspace{0.3em}

% Proof idea: Utilize the framework developed in Ewin Tang's work, difficulties with fermions (pairing)
\subsection{Technical highlights.}
We follow the framework developed for spin systems \cite{bakshi2024high}. Here we highlight the necessary modifications to accommodate the fermionic anti-commutation relations, which lead to an adaptive choice of the matching $M^{(i)}$ in the Gaussian states.

Similar to the qubit case \cite{bakshi2024high}, for any sequence of decreasing subsets of Majorana sites, that is, a sequence of subsets $S_{0}\supsetneq S_{1}...\supsetneq S_{2n}$ with $S_{0}=\{1,...,2n\}$ and $S_{2n}=\emptyset$, we define the restricted Hamiltonian $H^{S_i}$ to be the Hamiltonian consisting of terms of $H$ whose supports are entirely contained in $S_i$. 
This allows us to decompose the Gibbs states into product of pieces:
\begin{align}
e^{-\beta H} = M_{2n} \dots M_1 M_1^\dag...M_{2n}^\dag\,
   \text{ for }\, M_i :=  e^{\beta H^{S_{i}}/2}e^{-\beta H^{S_{i-1}}/2},
\end{align}
For easy illustrations, we temporarily assume that $S_i=\{i+1,...,2n\}$.

We further show that at high temperatures,  each \( M_i \) is a probability distribution over operators of a simple form.
At a high level, we use this simple form to randomly sample a simple operator from $M_1M_1^\dagger$, isolate (pin out) the first site $\gamma_1$, and move it to the left of $M_{2n} \dots M_2 M_2^\dag...M_{2n}^\dag$. This procedure leaves us with an operator that acts only on the remaining Majorana operators $[2,3...,2n]$,  allowing us to recursively apply similar techniques to isolate subsequent operators.

In the qubit setting, where $S_i$ corresponds to the index of qubits and $\gamma_1$ corresponds to an operator acting on qubit $1$, it is relatively easy to pin out  $\gamma_1$
since operators on qubit $1$ commute with the remaining terms \( M_{2n} \cdots M_2 \), which do not act on qubit $1$. However, this approach must be modified in the fermionic setting, where the basic building blocks are Majorana operators that anti-commute rather than commute. For this reason, sites must be pinned in pairs, since operators of the form  $\gamma_k \gamma_l$ do commute with other operators with disjoint supports and even weight.
Moreover, two sites can only be pinned out together once the remaining terms in $M_{2n} \dots M_{2n}^\dag$ do not include either site in their support. 
Pairs of sites are pinned whenever possible, forming a partial matching that is adaptively determined by the random sampling process.

Using the above ideas along with a few additional sampling steps and algebraic manipulations, a distribution of unnormalized fermionic Gaussian states can be obtained that reproduces $e^{-\beta H}$ in expectation. 
%\jold{However, these sampled fermionic Gaussian states do not have to be normalized, and worryingly their normalization constants could be negative. } 
However, it is particularly concerning that their normalization constants could be negative.

To complete the proof, we demonstrate that these normalization constants are actually all nonnegative. The proof of positivity is essential, and it is what separates this approach from a naive sampling of terms from $M_{2n}\dots M_{2n}^\dag$. 
The proof makes use of a nuance in the algorithm: the $S_i$ must be chosen adaptively according to a rule, rather than being set a priori.
Likewise, the pairs of pinned indices  $(k,l)$ are chosen carefully and adaptively. 
This latter requirement is a fundamental distinction from the spin system case, in which all the sampled states were product states with respect to the same tensor product decomposition.
In the fermionic case, each sampled state is a Gaussian state with respect to a different pairing of Majorana sites, and the state could therefore exhibit more complex correlations between sites than product states in a spin system. 
And just as in the original proof, the adaptive, algorithmic nature of the decomposition is indispensable in proving an entirely physical result. 
For a full description of the algorithm, see Section \ref{sec:structure_alg}.

To complement our structural results for bounded-degree local Hamiltonians, we also establish that the Gibbs states of the Sachdev–Ye–Kitaev (SYK) model are not mixtures of Gaussian states at sufficiently high temperatures. 
Our argument is based on comparing energies achieved by the Gibbs states and Gaussian states. 
In particular, we show that for an SYK Hamiltonian $H$, the Gibbs state $\rho_\beta$ achieves a better energy than any Gaussian state can. 
We achieve this by calculating the derivative of $\text{tr}(\rho_\beta H)$ with respect to $\beta$ at $\beta = 0$ and upper bounding the corresponding second derivative.
Combining this approach with known upper bounds for the energy of Gaussian states, we conclude that for $q \geq 12$, the SYK Gibbs state cannot be convex-Gaussian.

\subsection{Contribution and open problems} %\YT{added} 
In this work we prove that the high-temperature fermionic Gibbs state is a mixture of fermionic Gaussian states. By modifying the sampling procedure used in this proof, we also provide an efficient algorithm to sample normalized Gaussian state density matrices whose expectation value yields the target Gibbs state. Our result rigorously demonstrates the absence of the fermionic sign problem at high temperature, and sharpens the boundary between quantum advantage and classical simulability for fermionic systems.  As a contrasting example, we find that high-temperature Gibbs states of the Sachdev–Ye–Kitaev (SYK) model are not mixtures of Gaussian states, showing that the structural phenomena for bounded-degree fermionic Hamiltonians break down in strongly interacting systems.

To understand the limitations of simulating quantum systems on classical computers, one needs to look beyond exactly solvable models, and explore examples that are at the boundary between hard and easy systems. These examples include weakly interacting spin systems \cite{BravyiDiVicenzoLoss2008polynomial}, the quantum impurity model \cite{BravyiGosset2017complexity}, and quantum systems at high temperature. Recently, significant progress has been made in the rigorous analysis of the high-temperature regime, leading to provable algorithms for sampling from the Gibbs state \cite{YinLucas2023polynomial, bakshi2024high} and for computing the partition function in the presence of long-range interactions \cite{Sanchez2025high}. Our result represents a step towards better understanding of fermionic systems at high temperature and our techniques provide insights for generalizing results for spin systems to the fermionic setting.

\vspace{0.3em}

We can see that the results presented in this work open up numerous directions for future investigation.
Our results implies that high-temperature Gibbs states of geometrically local fermionic Hamiltonians—and more generally, bounded-degree local fermionic  Hamiltonians—is a probability distribution over  Gaussian states  and can be efficiently simulated on classical computers. It remains to be seen whether our techniques can be extended to rule out quantum advantage in other fermionic settings.
In particular, \cite{TongZhan2024fast,SmidMeisterBertaBondesan2025polynomial} showed that quantum computers can efficiently prepare  fermionic Gibbs states
 at any temperature when the non-quadratic interaction is sufficiently weak. An interesting open question is whether these Gibbs states can also be expressed as mixtures of Gaussian states, thereby enabling efficient classical simulation.  Alternatively, it may be shown that in this regime the Gibbs state is not a mixture of Gaussian states
 , in which case the algorithmic framework developed in \cite{bakshi2024high} would no longer apply, suggesting the possibility of quantum advantages.

%\YT{I added this paragraph about QMC and incorporated what Jiaqing wrote about the sign problem into it.}
It is also natural to consider how the techniques in this work and in \cite{bakshi2024high} relate to quantum Monte Carlo methods such as determinant quantum Monte Carlo and auxiliary-field quantum Monte Carlo \cite{blankenbecler1981monte,zhang1995constrained,zhang2003quantum,lee2022twenty}. A key difference in the algorithm we use is that we only sample positive semidefinite density matrices whose expectation value yields the Gibbs state, whereas quantum Monte Carlo methods typically sample wavefunctions, e.g., in the form of Slater determinants, together with weights that may have complex phases. This key difference is one reason why our algorithm does not suffer from the fermionic sign problem~\cite{ten1995proof,jiang2025local,santos2003introduction}. It is of interest to investigate whether the structural result of expressing fermionic Gibbs states as mixture of Gaussian states is useful for ameliorating the sign problem. The connection to Monte Carlo methods for stoquastic Hamiltonians \cite{Suzuki1977monte,bravyi2017complexity,Bravyi2014monte,Crosson2020classical,Takahashi2024rapidly}, which is typically free from the sign problem, is another direction to be explored.

%\jiaqing{I added this paragraph for the open question related to learning}
Additionally, it is interesting to explore whether our results could inspire new learning algorithms for fermionic systems. In classical machine learning, mixtures of Gaussian distributions are often learned using the Expectation–Maximization (EM) algorithm~\cite{bishop2006pattern,dasgupta2007probabilistic} —an iterative method which monotonously increases the likelihood and differs significantly from existing approaches for learning fermionic states~\cite{ni2024quantum,mirani2024learning,mele2025efficient}. Given our structural result that high-temperature fermionic Gibbs states can be expressed as mixtures of fermionic Gaussian states, it is natural to ask whether the EM algorithm can be adapted to efficiently learn such states.\\
%We can see that the results presented in this work open up numerous directions for future investigation.\\

%\jold{Our result implies that high-temperature Gibbs states for  geometrically local Hamiltonians are classically easy, but this is far from settling the question of quantum advantage in the fermionic setting. Recall that the quantum algorithms in \cite{TongZhan2024fast,SmidMeisterBertaBondesan2025polynomial} work at any temperature when the non-quadratic interaction is sufficiently weak. An interesting open question is, when the interaction strength is sufficiently weak but independent of the system size, whether the high-temperature Gibbs state is still a mixture of Gaussian states. Furthermore, can this lead to an efficient algorithm for weakly interacting fermions? Or can we show that in this setting the Gibbs state is not a mixture of Gaussian states and therefore the algorithm framework used in \cite{bakshi2024high} and this work cannot be used to solve the problem? We can see that the results obtained in this work open up numerous possibilities to be explored in future works.}

% \jiaqing{Sudden death of entanglement in the Fermionic system}
% \YT{I'll write this}

\subsection{Overview and organization} 
In this paper, we will construct two algorithms, a structural algorithm and a sampling algorithm.
The two algorithms rely on different, but both constant, temperature bounds.
The structural algorithm does not classically prepare the Gibbs state as a probability distribution of Gaussian states, but its correctness implies the structural result that the Gibbs state lies in the convex hull of Gaussian states. 

In particular, the structural algorithm outputs a PSD matrix of the form $\sigma = \frac{\lambda}{\text{tr}(I)}\prod_{(k, l) \in M} (I \pm i\gamma_k \gamma_l)$, where no index is contained in two elements of $M$ (i.e., $M$ is a partial matching of the indices). 
The matrices $\sigma$ need not be normalized, and can have varying traces $\lambda \geq 0$—they are unnormalized Gaussian states with \textit{positive} normalization constants. 
This algorithm, therefore, does not output quantum states at all. 
However, we will ensure that $\mathbb{E}[\sigma] = e^{-\beta H}$.
Despite not preparing normalized quantum states, the existence of this algorithm will algebraically imply that the Gibbs state is a convex combination of Gaussian states. 

Indeed, if the algorithm outputs a given $\sigma^{(i)}$ with probability $p_i$, then 
\[e^{-\beta H} = \sum_i p_i\sigma^{(i)}.\]
Each $\sigma^{(i)} = \lambda_i \rho^{(i)}$, for a Gaussian state $\rho^{(i)}$ and with $\lambda_i = \text{tr}(\sigma^{(i)})$.
Defining the nonnegative values
\[q_i = \frac{p_i\lambda_i }{\text{tr}(e^{-\beta H})}, \]
we obtain 
\[\rho_\beta= \sum_i q_i \rho^{(i)}.\]
Taking the trace of both sides of the above equation, we verify that $\sum_i q_i = 1$, proving that $\rho_\beta$ is indeed a convex combination of Gaussian states.

After constructing this structural algorithm in Section \ref{sec:structure_alg}, we will adapt it to an efficient sampling algorithm in Section \ref{sec:sampling_alg} that can instead efficiently sample $\rho^{(i)}$ with probability $q_i$. \\

\noindent \textbf{Organization.} The rest of the paper is organized as follows: In Section~\ref{sec:problem_setup} we provide the necessary definitions of Majorana operators and the Hamiltonians that we are going to study. In Section~\ref{sec:telescoping} we break the Gibbs state into a product of factors, each of which can be efficiently sampled. In Section~\ref{sec:structure_alg} we present a sampling algorithm over unnormalized fermionic Gaussian states that produces the Gibbs state in expectation value, which proves the structural result that the Gibbs state is a mixture of Gaussian states. In Section~\ref{sec:sampling_alg} we modify the algorithm in the previous section to obtain an efficient algorithm that outputs random normalized fermionic Gaussian states whose expectation value is the Gibbs state.
%\YC
Lastly, in Section~\ref{sec:SYK_non-gaussian} we present the complementary result that high-temperature Gibbs state of the dense SYK model is not a mixture of Gaussian states.

\section{Problem Setup}
\label{sec:problem_setup}
We use $[2n]$ to denote the set $\{1,2,...,2n\}$. 
For a subset $S\subseteq [2n]$, we use $|S|$ to denote the number of elements in $S$, and $S^c$ to denote the complement set $[2n]\backslash S.$ 
For two subsets $A,B\subseteq [2n]$, we use the notation $A-B = \{a \in A | \,a \notin B\}$, so that the operation is well-defined even when $B\not\subseteq A$.

A system of $n$ fermion modes can be described by $2n$
\textit{Majorana (fermion) operators} $\gamma_j$, 
%\jiaqing{I change Majorana fermion operators to Majorana (fermion) operators, since in the following paragraghs we often use `` Majorana operators."}
$j\in [2n],$  which satisfy 
$$\gamma_j^2 = I, \gamma_j=\gamma_j^\dagger, tr(\gamma_j)=0 \text{ and }\gamma_i\gamma_j=-\gamma_j\gamma_i, \forall i\neq j.$$ 
We use Majorana string to refer to a product of an even number of distinct Majorana operators, that is, $G = \alpha \gamma_{i_1} \cdots \gamma_{i_k}$, with a complex coefficient $\alpha$ and even $k$.
We define its support to be $\supp(G):=\{i_1,...,i_k\}$. 
%\jold{We use the term \textit{$\mathcal{R}$-local Majorana string} to refer to  the product of $\mathcal{R}$ distinct Majorana fermion operators. We always assume $\mathcal{R}$ is even.}\jiaqing{A 4-local Hamiltonian show allows both quardratic and quartic terms.}
We define $\mathcal{M}$ to be the set of all $\mathcal{R}$-local Majorana strings with coefficient $i^{\mathcal{R}/2}$, which is used to ensure that the term is Hermitian, that is, $\left( i^{\mathcal{R}/2}\gamma_{i_1}... \gamma_{i_\mathcal{R}}\right)^\dagger  = i^{\mathcal{R}/2}\gamma_{i_1}... \gamma_{i_\mathcal{R}}$. 
Moreover, we use $\mathcal{M}^*$ to denote the set of all Majorana strings with coefficient in $\{\pm 1, \pm i\}$ (so either Hermitian or anti-Hermitian Majorana strings).

We say a fermionic Hamiltonian $H$ on $2n$ Majorana operators is of locality $\mathcal{R}$ and degree $d$, if
\[
   H = \sum_{a=1}^m \lambda_a G_a, 
\]
where $0\leq \lambda_a \leq 1$,  %\YC{$0 \leq \lambda_a \leq 1$?} \jold{$G_a\in \mathcal{M}$ is $\mathcal{R}$-local, i.e $|\text{supp}(G_a)| \leq \mathcal{R}$,} 
$|\text{supp}(G_a)|$ is even and  $\leq \mathcal{R}$, and   any local term $G_a$ has overlapping support with at most $d$ other terms, i.e. $|\{b \,| \,\text{supp }(G_a) \cap \text{supp}(G_b) \neq \emptyset\}| \leq d$. 
%each Majorana operator  $\gamma_j$ appears in at most  $d$ different local terms $G_a$. We implicitly assume that $0\leq \lambda_a \leq 1$ and $m=poly(n)$.

Notably, fermionic Gaussian states are the ground states and thermal states of $2$-local fermionic Hamiltonians,
%$2$-local fermionic Hamiltonian corresponds to non-interacting fermions and is a fermionic analogy of $1$-local qubit Hamiltonian, thus Gaussian states are the natural fermionic analogy of product states. 
and $2$-local fermionic Hamiltonians describe non-interacting fermions and serve as the fermionic counterpart to $1$-local qubit Hamiltonians.
Consequently, fermionic Gaussian states can be seen as the natural fermionic analogue of product states.

A special case of  Gaussian states can be characterized by a partial matching $M$ of the $2n$ Majorana operators, specified by $m\leq n$ disjoint pairs $(k,l)$ and is defined as
\begin{align}
	\rho :=\frac{1}{2^n}\prod_{(k,l)\in M} \left(I+ i \gamma_{k}\gamma_{l}\right).
\end{align}
We use the term unnormalized Gaussian state for an operator to describe a Gaussian state with a positive proportionality constant.

Lastly, we comment on the norms used in the proofs. We denote the operator norm by $\lVert \cdot \rVert_\infty$, and the $L_1$ distance of probability distributions and the Schatten 1-norm of operators by $\lVert \cdot \rVert_1$.

\section{Telescoping and Sampling}\label{sec:telescoping}
%\jiaqing{I have changed this ``subsection" to ``section".}

The algorithm will leverage various telescoping expansions of $e^{-\beta H}$. 
For any subset $S\subseteq [2n]$, define the restricted Hamiltonian $H^S$ as the  Hamiltonian  consisting of the local terms of $H$ whose supports are entirely contained in $S$.
%\jold{Define $H^S$, for any subset $S$, to consist of the local terms of $H$ for which the support is contained in $S$.} 
Set $S_0=[2n]$ and $S_{2n}=\emptyset$.
For a given sequence of properly decreasing subsets $S_0\supsetneq \dots \supsetneq S_{2n}$ with $|S_{i}| -|S_{i+1}|=1$,
%\jiaqing{$S_n$ to $S_{2n}$ and add $\supset$.}
% \jold{For a given sequence of properly decreasing subsets  $S_0, \dots, S_{n}$ of indices}, 
we define
\begin{align}
    M_i =  e^{\beta H^{S_{i}}/2}e^{-\beta H^{S_{i-1}}/2},\label{eq:Mi}
\end{align}
% \jiaqing{If we use ``note that", we should define it first. I feel $[2n]$ or $\{1,...,2n\}$ is more standard than $[1,...,2n]$.}
Note that $S_0=[2n]$ 
%\jold{$S_0 = [1, \dots, 2n]$} 
and $S_{2n} =\emptyset$, so $H^{S_{2n}} = 0$ and $H^{S_0} = H$.  
%\jold{We may therefore write} 
We can therefore write
\begin{align}
    e^{-\beta H} = M_{2n} \dots M_1 M_1^\dag \dots M_{2n}^\dag.
\end{align}
%\jold{This decomposition of $e^{-\beta H}$ breaks it up into small pieces, each of which has a well-behaved Taylor decomposition.}
The above equation breaks $e^{-\beta H}$ into small pieces. 
Below, we will provide analysis on the behaviors of this decomposition that are important for our later structural and sampling algorithms.  
In Lemma \ref{lem:fermion_taylordecomp}, adapted from Theorem 4.2 in \cite{bakshi2024high}, we show that each of $M_i$ has a well-behaved Taylor decomposition, where the coefficients decay exponentially in the degree.
%\jiaqing{Is the following Lemma very similar to Ewin's proof. If so, we should mention it.} 
Corollary \ref{cor:sampl_cjt} builds on this by showing that the terms of a fixed degree in the expansion can be efficiently sampled and their support is bounded, which is useful for the later Gibbs sampling algorithm in Section \ref{subsec:truncated_sample_tree}.
Lastly, we will show in Lemma \ref{lem:fermion_convexcomb} that each update step $M_i$ in the telescoping expansion can be written as a convex combination of simple, local perturbations of the identity. 
Lemma \ref{lem:fermion_convexcomb}, a consequence of the Taylor decomposition is at the heart of the structural algorithm, because it enables $e^{-\beta H}$ to be decomposed as a convex combination of simple operators.

\begin{lemma}\label{lem:fermion_taylordecomp}
  For any $S_{i}$ and $S_{i-1}$, if $\beta \leq \frac{1}{Cd}$ for some value $C > 1$, then the corresponding $M_i$ defined in Eq.~(\ref{eq:Mi})  can be decomposed as $I + \sum_{t >0} \sum_j c_{j,t} \Gamma_{j,t}$, where the following properties hold:
    \begin{itemize}
        \item $\Gamma_{j,t}\in \mathcal{M}^*$ is $\mathcal{R}t$-local Majorana string with $\supp(\Gamma_{j,t})\subseteq S_{i-1}$.
        %\jiaqing{It's better to summarize the properties of $\Gamma_{j,t}$ in one Lemma. I have also written a one-sentence proof in this lemma.}
        \item  $c_{j,t} \geq 0$ and  $\sum_j c_{j,t} \leq C^{-t}.$ 
    \end{itemize}
    % \jiaqing{In the problem setup we use Majorana operator for $\gamma_j$, and Majorana string for products of $\gamma_j$. Please unify the notations.}
    % \jold{For any $S_{i}$ and $S_{i-1}$, if $\beta \leq \frac{1}{Cd}$ for some value $C > 1$, then the corresponding $M_i$ as above can be decomposed as $I + \sum_{t >0} \sum_j c_{j,t} \Gamma_{j,t}$, where  $\Gamma_{j,t}\in \mathcal{M}^*$ are $\mathcal{R}t$-local Majorana operators, $c_{j,t} \geq 0$, and such that $\sum_j c_{j,t} \leq C^{-t}.$}
\end{lemma}

\begin{proof}
    We can write 
    \begin{equation}\label{eq:Mi_expansion}
        \begin{aligned}
        e^{\beta H^{S_{i}}/2}e^{-\beta H^{S_{i-1}}/2} &= \left(\sum_{k=0}^{\infty}\frac{(\frac{\beta}{2})^{k}(H^{S_{i}})^k }{k!}\right) \left(\sum_{l=0}^{\infty} \frac{(\frac{\beta}{2})^{l} (-H^{S_{i-1}})^{l}}{l!}\right) \\
        &= \sum_{t=0}^{\infty} \frac{\beta^t}{2^t t!} \sum_{k=0}^{t} \frac{(H^{S_{i}})^k (-H^{S_{i-1}})^{t-k} t!}{k! (t-k)!} \\
        &= \sum_{t=0}^{\infty} \frac{\beta^t}{2^t t!} \underbrace{\sum_{k=0}^{t} {t \choose k}(H^{S_{i}})^k (-H^{S_{i-1}})^{t-k}}_{f_t (H^{S_{i-1}}, H^{S_{i}})}
        \end{aligned}
    \end{equation}
    Here, the leading term $f_0 (H^{S_{i-1}}, H^{S_{i}}) = I$, and the subsequent terms follow the recursive relationship 
    \begin{equation}\label{eq:recursive}
        \begin{aligned}
            f_t (H^{S_{i-1}}, H^{S_{i}}) &= H^{S_{i}} f_{t-1} (H^{S_{i-1}}, H^{S_{i}}) - f_{t-1}(H^{S_{i-1}} , H^{S_{i}}) H^{S_{i-1}} \\
            &= [H^{S_{i}}, f_{t-1} (H^{S_{i-1}} , H^{S_{i}})] - f_{t-1} (H^{S_{i-1}} , H^{S_{i}}) (H^{S_{i-1}} - H^{S_{i}})
        \end{aligned}
    \end{equation}
    The polynomial $p_t(H^{S_{i-1}}, H^{S_{i}}) := \frac{\beta^t}{2^t t!} f_{t} (H^{S_{i-1}} , H^{S_{i}})$ can be expanded as a nonnegative combinations of Majorana strings in $\mathcal{M}^*$. We will do so by using the recursive relations of $f_t$, and show by induction that the desired properties are satisfied for this decomposition.
    In particular, we will show that each $\Gamma_{j,t}$ is a $t$-fold product of operators $G_a$. 
    % \jiaqing{it seems in other places we use $t$-fold. So I unify the notations.}
    It will follow that $\Gamma_{j,t}$ is $\mathcal{R}t$-local, since each term of the Hamiltonian acts on at most $\mathcal{R}$ sites.

Consider the expansion in Eq.~(\ref{eq:Mi_expansion}). We begin with the base case $t = 0$. Since $f_0 (H^{S_{i-1}}, H^{S_{i}}) = I$, we define $c_{0,0}=1$, $T_{0,0}=I$ and $c_{j,0}=0, T_{j,0}=0$ $\forall j\geq 1$, given that the identity operator trivially satisfies the locality requirement with coefficient $1$. 
% \jiaqing{I think we should delete the usage of ``may"... Also we should define $c_{j,0}$}
    % \jold{We  may take the base case to be $t = 0$, treating $I$ as $c_{0, 0}\Gamma_{0, 0}$ with $c_{0, 0}=1$ and $\Gamma_{0, 0} = I$. 
    % The identity operator trivially satisfies the locality requirement with coefficient $1$.} 
    For the inductive case, assume for some  
    % \jold{$t > 0$},\jiaqing{our base case is $t=0$} 
    $t\geq 0$, we have \[p_t(H^{S_{i-1}}, H^{S_{i}}) = \sum_{j} c_{j,t} \Gamma_{j,t}\] 
    with the desired properties. 
    That is, $\Gamma_{j,t}$ is a $t$-fold product of operators $G_a$, so $\mathcal{R}t$-local, $c_{j,t} \geq 0$, and $\sum_j c_{j,t} \leq C^{-t}.$
    Directly applying the recursive relationship in Eq.~(\ref{eq:recursive}), we have that 
    \begin{align}
        p_{t+1} (H^{S_{i-1}}, H^{S_{i}}) = \frac{\beta}{2(t+1)} \sum_{j} c_{j,t} \left([H^{S_{i}}, \Gamma_{j,t}] - \Gamma_{j,t} \left(H^{S_{i-1}}-H^{S_{i}}\right)\right).\label{eq:pt}
    \end{align}

 % \jold{Note that the commutator $[H^{S_{i}}, \Gamma_{j,t}]$ is non-zero only on the local terms where $H^{S_{i}}$ and $\Gamma_{j,t}$ overlap, in particular $G_a$ for $a$ in
 %    \[T_{1,j} = \{b : [G_b, \Gamma_{j,t}] \neq 0\}.\]  \jiaqing{$\supp(G_b)\subseteq S_i$}
 %    % \YT{this $T_1$ depends on $j$. Since later summation over $j$ is involved we should write $T_{1,j}$ to make this dependence explicit.}
 %    So we can write the commutator 
 %    \begin{align*}
 %    [H^{S_{i}}, \Gamma_{j,t}] &= \sum_{a} (2\lambda_a) \frac{[G_a, \Gamma_{j,t}]}{2}\jiaqing{a\in T_{1,j}}\\
 %    &= \sum_{a \in T_{1,j}} (2\lambda_a) G_a\Gamma_{j,t}.
 %       \end{align*}
 % }      

 In the following, we will simplify Eq.~(\ref{eq:pt}) by observing that the commutator $[H^{S_{i}}, \Gamma_{j,t}]$ is non-zero only on the local terms where $H^{S_{i}}$ and $\Gamma_{j,t}$ overlap. 
 Define $T_{1,j}$ as the set of local terms of $H^{S_i}$ whose support overlaps with $\Gamma_{j,t}$: 
\[T_{1,j} := \{b : \supp(G_b)\subseteq S_i \text{ and } [G_b, \Gamma_{j,t}] \neq 0\}.\] 
Note that 
\begin{align*}
    [H^{S_{i}}, \Gamma_{j,t}] &= \sum_{a\in T_{1,j}} (2\lambda_a) \frac{[G_a, \Gamma_{j,t}]}{2}\\
    &= \sum_{a \in T_{1,j}} (2\lambda_a) G_a\Gamma_{j,t},
\end{align*}
since the expression $\frac{[G_a, \Gamma_{j,t}]}{2}$ either equals zero, when the two terms commute, or equal to $G_a \Gamma_{j,t}$ when they anti-commute. 

Moreover, recall that by definition $S_{i-1}$ has just one more element than $S_i$. 
Then, the local $G_b$ in $H^{S_{i-1}}-H^{S_i}$ can be characterized by the following set
    \[T_2 := \{b : v \in \supp(G_b) \text{ for } \{v\} = S_{i-1} - S_{i} \text{ and } G_b \in H^{S_{i-1}}\}.\]
    % \YT{$G_b$ above also needs to be contained in $H^{S_{i-1}}$?}
We have therefore obtained a decomposition
    \begin{equation} \label{eq:inductive_c}p_{t+1}(H^{S_{i-1}}, H^{S_i}) =\frac{\beta}{2(t+1)}\left( \sum_{ a \in T_{1, j}, j} (2\lambda_a c_{j,t})(G_a\Gamma_{j,t}) + \sum_{j, a \in T_2} (\lambda_a c_{j,t}) (-\Gamma_{j,t} G_a)\right).
    \end{equation}
Further note that by induction, every Majorana string in Eq.~(\ref{eq:inductive_c}) is a $(t+1)$-fold product of terms $G_a$ and is thus at most $\mathcal{R}(t+1)$-local. Moreover, note that $\supp(\Gamma_{j,t})\subseteq S_{i-1}$ since both $H^{S_i}$ and $H^{S_{i-1}}$ are supported on $S_{i-1}.$ 

% \jold{
%     Next, note that $H^{S_{i-1}}-H^{S_{i}}$ consists of Majorana terms that act on $S_{i-1}$ but not on $S_i$. \jiaqing{This sentense is a bit vague.}
%     Since $S_{i-1}$ has just one more index than $S_i$, these are the terms $G_a$ for $a$ in 
%     \YT{here we are assuming $S_{i-1}=S_i\cup \{v\}$, which should be stated explicitly when we introduce $S_i$} \jiaqing{I have done this.}
%     \[T_2 = \{b : v \in \supp(G_b) \text{ for } \{v\} = S_{i-1} - S_{i} \text{ and } G_b \in H^{S_{i-1}}\}.\]
%     % \YT{$G_b$ above also needs to be contained in $H^{S_{i-1}}$?}
%     We have therefore obtained a decomposition
%     \begin{equation} \label{eq:inductive_c}p_{t+1}(H^{S_{i-1}}, H^{S_i}) =\frac{\beta}{2(t+1)}\left( \sum_{j, a \in T_{1, j}} (2\lambda_a c_{j,t})(G_a\Gamma_{j,t}) + \sum_{j, a \in T_2} (\lambda_a c_{j,t}) (-\Gamma_{j,t} G_a)\right),
%     \end{equation}
%     \jiaqing{$j,t \in T_{i,j}$?}
%     in which by induction we see that every Majorana term is a $(t+1)$-fold product of terms $G_a$ and is thus at most $\mathcal{R}(t+t)$-local.\jiaqing{I have modified this.}
    
% }

    Finally, we prove that for this decomposition, $\sum_{j'} c_{j',(t+1)} \leq C^{-(t+1)}$.
    To do this, we will bound the sizes of $T_{1, j},\forall j$ and the size of $T_2$.
    First, note that for any $\Gamma_{j,t}$, which is a  $t$-fold product of Majorana strings $G_a$, there are at most $td$ operators $G_a$ in the Hamiltonian that have a nonzero commutator with $\Gamma_{j,t}$.
    % \jold{we note the general fact  a $t$-fold product of Majorana strings $G_a$ has nonzero commutator with at most $td$ other terms in the Hamiltonian.} 
    It follows that $|T_{1, j}| \leq td$. 
    On the other hand, each $a \in T_2$ corresponds to a $G_a$ whose support includes the one index $v \in S_{i-1} - S_i$. All of these operators have intersecting support, which means $|T_2| \leq d+1$. 
    We may now write
    \begin{equation}
        \begin{aligned}
        \sum_{j'} c_{j',(t+1)} &= \frac{\beta}{2(t+1)}\sum_{a \in T_{1,j}}\sum_{j} 2\lambda_a c_{j,t} + \frac{\beta}{2(t+1)}\sum_{a \in T_2}\sum_j \lambda_a c_{j,t}\\&\leq \frac{\beta}{2(t+1)} \sum_{j} c_{j,t} (2td+d+1)\\ &\leq \frac{\beta}{2(t+1)} \sum_{j} 2(t+1)c_{j,t} d\\&\leq \frac{1}{C}\sum_j c_{j,t}.
        \end{aligned} 
    \end{equation}
    where the first inequality uses that $d \geq 1$ and the last inequality uses that $\beta \leq \frac{1}{Cd}$. Combining with the inductive hypothesis that $\sum_{j} c_{j,t} \leq C^{-t}$, we arrive at the desired bound.

\end{proof}

\noindent
%\jiaqing{I think we shouldn't mention Lemma \ref{lem:eff_sampl} here, since the readers don't know what it is anyway.}
%Lemma \ref{lem:fermion_taylordecomp} is crucial to proving a sampling lemma, stated in Lemma \ref{lem:eff_sampl}.

\begin{lemma}\label{lem:fermion_convexcomb}
%\jold{We assume  $\beta \leq \frac{1}{2N \mathcal{R} d}$. Each} 
Assume that $\beta \leq \frac{1}{2N \mathcal{R} d}$ for some  $N>0$, then each operator $$M_i = e^{\beta H^{S_i}/2}e^{-\beta H^{S_{i-1}}/2}$$ and correspondingly $M_i^\dag$ is a probability distribution over operators of the form $I+(N\mathcal{R})^{-t_k}\Gamma_k$, where for any $k$, $\Gamma_k \in \mathcal{M}^*  \cup \{0\}$ is $t_k\mathcal{R}$-local and is supported on $S_{i-1}$. 
%\jiaqing{in other proofs, we use  ``probability distribution'' rather than only ``convex combinations".}
\end{lemma}
\begin{proof}
We know from Lemma \ref{lem:fermion_taylordecomp} that $M_i$ can be written as $I + \sum_{t>0} \sum_j c_{j,t} \Gamma_{j,t}$. Then we can rewrite $M_i$ as 
\begin{align}
    M_i &= \left(1- \sum_{t>0}\sum_j c_{j,t}(N\mathcal{R})^t\right) I  + \sum_{t>0}\sum_j \left(c_{j,t}(N\mathcal{R})^t\right)\left(I+ (N\mathcal{R})^{-t} \Gamma_{j,t})\right)
    .\label{eq:Mi_convex}
    \end{align}
    Note that by Lemma \ref{lem:fermion_taylordecomp} we know that $\sum_{j} c_{j,t}\leq (2N\mathcal{R})^{-t}$, $c_{j,t}\geq 0$,
    thus in Eq.~(\ref{eq:Mi_convex}) the coefficients of $I$ and $\left(I+ (N\mathcal{R})^{-t} \Gamma_{j,t})\right)$ are non-negative and sum to one. Therefore, Eq.~(\ref{eq:Mi_convex}) shows that $M_i$ is a probability distribution over operators $I+(N\mathcal{R})^{-t_k}\Gamma_k$, where to ease notation we treat $I=I+(N\mathcal{R})^{0}0$ as the term $k=0$, and we relabel the indexes $(j,t)$ for $t>0$ arbitrarily as $k=1,2,...$.

Moreover, by Lemma \ref{lem:fermion_taylordecomp} we know that for $t>0$,  $\Gamma_{j,t}$ is $\mathcal{R}t$-local and is supported on $S_{i-1}$, thus the desired properties hold.  The desired properties  trivially holds for the term $I=I+(N\mathcal{R})^{0}0$. 
Thus we complete the proof, and the same argumentation applies almost identically to $M_i^\dag$. 

\end{proof}

%\jiaqing{TBC: in the intro and overview, mention ``structural result" and ``sampling result" explicitly.}
The above lemma will be used with $N = 24$ to prove the structural result, and with $N = 25$ to prove the sampling result. 
This is where the bounds on inverse temperature $\beta$ arise from.

% \jiaqing{To akshar: Move the Corollary 4 after Lemma 5 and mention that this lemma will not be used until section 6.}
Lastly, we record a corollary which will be useful for the later Gibbs sampling algorithm in Section \ref{sec:sampling_alg}. 
\begin{corollary}\label{cor:sampl_cjt}
For any fixed $t$ and $\beta \leq \frac{1}{Cd}$, the values $(c_{j,t}, \Gamma_{j,t})$ can be sampled in time $\text{poly}(t, n)$ according to a distribution $\mu$ such that $\mu_{j, t} \geq c_{j, t} C^t $, and which has support of size at most $(t+1)!(d+1)^t$, where $d$ is the degree of the interaction graph of $H$.
\end{corollary}
\begin{proof}
The distribution $\mu$ will be sampled by a repeated application of Equation \ref{eq:inductive_c}. 
We will prove this result by induction. 
For $t = 0$, the only term 
in the decomposition is the identity operator $I$ with coefficient $c_{0,0} = 1$, corresponding to $\Gamma_{0,0} = I.$
We define $\mu_{0,0} = 1.$
Thus, 
\[\mu_{0,0} = c_{0,0} \geq c_{0,0} C^0,\]
and the support size is $1$.
The base case holds. 

Suppose that for some $t \geq 0,$ there exists a distribution $\mu$ over pairs $(c_{j,t}, \Gamma_{j,t})$ that can be sampled in $poly(t,n)$ time such that $\mu_{j,t} \geq c_{j,t} C^t$ for all $j$ and the support size is at most $(t+1)! (d+1)^t$.
We will construct such a distribution $\mu'$ for degree $t+1.$
Given a sampled pair $(c_{j,t}, \Gamma_{j,t})$ from $\mu$ at degree $t$, we extend it to a term $(c_{j',t+1}, \Gamma_{j', t+1})$ at degree $t+1$ by selecting some $a \in T_{1, j}$, or some $a \in T_2$.
This yields
\[c_{j', t+1} = \frac{\beta}{t+1}\lambda_ac_{j, t}\]
or
\[c_{j', t+1} = \frac{\beta}{2(t+1)}\lambda_ac_{j, t}\]
% Now, given $c_{j, t}$, $c_{j, t+1}$ may be chosen by selecting some $a \in T_{1, j}$, or some $a \in T_2$. 
% This yields
% \[c_{j', t+1} = \frac{\beta}{t+1}\lambda_ac_{j, t}\]
% or
% \[c_{j', t+1} = \frac{\beta}{2(t+1)}\lambda_ac_{j, t}\]
% \YT{this is not necessarily true since $c_{j,t+1}$ may come from $c_{j',t}$ for $j'$ not necessarily equal to $j$?}
respectively, with $\Gamma_{j', t}$ accordingly set to either $G_a\Gamma_{j,t}$ or $-\Gamma_{j,t}G_a$.
Now, $\mu_{j', t+1}$ is the distribution obtained by first sampling $(c_{j, t}, \Gamma_{j,t})$ from $\mu$, and then sampling a random value uniformly from $T_{1, j}$ with probability $P =\frac{2td}{2td + (d+1)}$ or from $T_{2}$ with probability $1-P = \frac{d+1}{2td + (d+1)}$. 
Note that $P \geq \frac{td}{(t+1)d}$ and $1-P \geq \frac{d+1}{2(t+1)d}$. 
Now, if $(j, t+1)$ corresponds to a value chosen from $T_{1, j}$, we have
\begin{align*}
\mu_{j', t+1} &\geq \frac{P}{td} \mu_{j, t}
\geq \frac{\beta C}{t+1}\mu_{j, t}
\geq C^{t+1}\frac{\beta}{t+1}c_{j, t}
\geq C^{t+1}c_{j', t+1},
\end{align*}
where first inequality comes from $|T_{1, j}| \leq dt$, the second inequality uses $1/d\geq \beta C$, the third inequality comes from the inductive hypothesis, and the last inequality follows from the fact that $0 \leq \lambda_a \leq 1$.

% \YT{should be
% \begin{align*}
% \mu_{j, t+1} &\geq \frac{P}{dt} \cdot \mu_{j, t}
% \geq \frac{\beta C}{t+1}\mu_{j, t}
% \geq C^{t+1}\frac{\beta}{t+1}c_{j, t}
% \geq C^{t+1}c_{j, t+1},
% \end{align*}
% where the second inequality uses $1/d\geq C\beta$?
% }
On the other hand, if $(c_{j', t+1}, \Gamma_{j', t+1})$ corresponds to a value chosen from $T_{2}$, then using the bound $|T_2| \leq d+1$, we obtain 
\begin{align*}
\mu_{j', t+1} &\geq \frac{1-P}{d+1} \mu_{j, t}
\geq \frac{1}{2(t+1)d}\mu_{j, t}\geq \frac{\beta C}{2(t+1)}\mu_{j, t}\geq C^{t+1}\frac{\beta}{2(t+1)}c_{j, t}\geq C^{t+1}c_{j', t+1},
\end{align*}
as desired.
% \YT{
% should be
% \begin{align*}
% \mu_{j, t+1} &\geq \frac{1-P}{d+1} \cdot \mu_{j, t}
% \geq \frac{1}{2(t+1)d}\mu_{j, t}\geq \frac{\beta C}{2(t+1)}\mu_{j, t}\geq C^{t+1}\frac{\beta}{2(t+1)}c_{j, t}\geq C^{t+1}c_{j, t+1}?
% \end{align*}
% }

Note that every step in the above runs in time that is polynomial in $n$ except potentially the uniform sampling of $T_{1, j}$ and $T_2$. 
However, this step can also be done efficiently. 
Specifically, the set $T_2$ contains at most $d+1$ elements, which can be easily enumerated from an encoding of the terms in $H$.
The set $T_{1,j}$ contains at most $td$ elements, corresponding to local Hamiltonian terms that overlap with the support of $\Gamma_{j,t}$.
Each element can be identified in time $O(t)$ by checking if it anti-commutes with   $\Gamma_{j,t}$.
% However, the at most $d+1$ elements of $T_2$ are easy to obtain from an encoding of the terms of $H$, and the at most $td$ elements of $T_{1, j}$ can be obtained in time $O(t)$ by checking if each of element that overlaps the support of $\Gamma_{j,t}$ anticommutes with the operator. 
Once these elements are obtained, sampling uniformly from them is efficient. 

We can also analyze the support of the sampling distribution. For each $t' < t$, there is a choice of one of at most $(t'+1)d+1 \leq (t'+1)(d+1)$ indices. 
For $t$ itself, it follows that there are at most $(t+1)!(d+1)^t$ elements in the support. 
\end{proof}

\section{Structural Algorithm}\label{sec:structure_alg} 
% \jiaqing{To Akshar/Yiyi: write a  sentence here like ``In this section, we show that [main goal/conclusion of this section]. 
% In the following, we will first give an overview of xxx. The detailed description of the Algorithm is in Section xxx. The formal proof is in xxx.
% }
In this section, we show that above a high temperature threhold that is constant of system size, the Gibbs state of fermionic local Hamiltonians is exactly a probability distribution over Gaussian states.
We first provide an overview in Section \ref{subsec:structural_overview}, then introduce the adaptive algorithm used to prove the structural result in Section \ref{subsec:struc_alg_description}. 
The formal proof of correctness of the algorithm is presented in Section \ref{subsec:structural_analysis}.

\subsection{Overview}\label{subsec:structural_overview}

The algorithm below stochastically prepares an unnormalized fermionic Gaussian state $\sigma$, such that the output distribution has expectation $\mathbb{E}[\sigma] = e^{-\beta H}$.
To do so, it iteratively expands $e^{-\beta H}$ using a stochastically chosen telescoping set of operators $M_j$, as in Section \ref{sec:telescoping}. 
In particular, at each step $j$, the algorithm chooses a site $i$ to remove from $S_{j}$, and defines $M_{j+1} \coloneq e^{\beta H^{S_{j+1}}}e^{-\beta H^{S_{j}}}$ with $S_{j+1} = S_j - \{i\}$. 

Using this telescoping product, the algorithm iteratively prepares operators $\Gamma_j \in \mathcal{M}, \sigma_j$, and a scalar $\alpha_j \geq 0$ at each step $j$, that satisfy the rule
\begin{equation}\label{eq:inductive_hypothesis}
\mathbb{E}\left[M_1^{-1}\dots M_{j}^{-1}(I+\alpha_{j}\Gamma_{j})\sigma_{j}(M_{j}^{\dag})^{-1} \dots (M_{1}^{\dag})^{-1}\right] = I,
\end{equation}
taking an expectation over the randomness of the algorithm. 
It is always true that $\sigma_j$ is an unnormalized separable state.
Moreover, the variables are initialized with $\sigma_0 = I$, $\Gamma_0 = 0$, and $\alpha_0 = 0$, while it is importantly ensured that $\Gamma_{2n} = 0$. 
This latter condition means that by the final step $j = 2n$, we have $M_{2n}M_{2n-1}\cdots M_{1}=e^{-\beta H/2}$, and the formula reads
\[\mathbb{E}\left[e^{\beta H/2}\sigma_{2n}e^{\beta H/2}\right] = I,\]
from which it follows that $\mathbb{E}[\sigma_{2n}] = e^{-\beta H}$.

How does this decomposition work? 
The algorithm's goal is to build up a distribution of separable states, $\sigma_{2n}$, that on expectation is equal to the Gibbs state.
It does this by iteratively constructing operators $(I + \alpha_j\Gamma_j)\sigma_j$, where $I+\alpha_j\Gamma_j$ is supported on sites that have yet to be extracted, and $\sigma_j$ is a separable state on the sites that have already been extracted.
We introduce the terminology of ``pinnable'' and ``pinned'' sites. 
A site $i$ is termed ``pinnable'' when $i \notin S_j^c$, and it is termed ``pinned'' when $i \notin S_j^c$ and $i \notin \Gamma_j$, so it has already been extracted from $\Gamma_j$.

It is crucial that any site extracted to $\sigma_j$ have support contained in $S_j^c$, because this ensures it has disjoint support with every subsequent $M_k$, $k > j$.
This is why only pinnable sites should be extracted to $\sigma_j$. 
When a site has been extracted to $\sigma_j$, or if it is naturally absent from the support of $\Gamma_j$ and not contained in $S_j$, it is ``pinned''. 
In either case, it is fully extracted, and is ensured to never reappear in subsequent $M_k$ operators. 

Let us see an example. 
Since $\sigma_0 = I$ and $\Gamma_0 = 0$, our algorithm begins with the operator $(I + \alpha_0\Gamma_0)\sigma_0 = I$. 
Then, in the first step, it conjugates this operator, in expectation value, by a well-chosen $M_1$ (for some site $i$), yielding 
\[M_1(I + \alpha_0\Gamma_0)\sigma_0M_1^\dag.\]
We may bring $\sigma_0$ to the side, since it commutes with $M_1$.
Notice that by the sampling argument in Lemma \ref{lem:fermion_convexcomb}, we can also sample Majorana operators $\beta_1\Lambda_1$ and $\beta_2\Lambda_2$ from $M_1$ and $M_1^\dag$ here, to obtain an expression 
\begin{equation}\label{eq:prod}
I + \alpha_1\Gamma_1 = (I + \beta_1\Lambda_1)(I + \alpha_0\Gamma_0)(I + \beta_2 \Lambda_2)\sigma_0
\end{equation}
for which $\mathbb{E}[M_1^{-1}(I + \alpha_1\Gamma_1)\sigma_1(M_1^\dag)^{-1}] = I$, where $\sigma_1 = \sigma_0$.
To simplify this expression further, however, we must perform additional sampling. 
Our expression can be foiled out as
\[c_1(I + c_1^{-1}\beta_1 \Lambda_1)\sigma_0 + \dots + c_6(I + (c_6^{-1}\beta_1\alpha_0\beta_2)(\Lambda_1\Gamma_j\Lambda_2)\sigma_0)\]
where $c_1, \dots, c_7 \geq 0$ are any values such that $\sum_i c_i = 1$. 
This expansion may be considered a ``sampling'' of terms from Equation \ref{eq:prod} according to a certain probability distribution. Notice that there is a choice of a probability distribution here, that is specified carefully in the algorithm. 
Sampling one of the above 7 terms according to the distribution $c_i$, we obtain a new expression of the form $(I + \alpha_1 \Gamma_1)\sigma_0$, where again $\mathbb{E}[M_1^{-1}(I + \alpha_1\Gamma_1)\sigma_0(M_1^\dag)^{-1}] = I$. 
Now, since $\Gamma_1 \in \mathcal{M}^*$, it is Hermitian of anti-Hermitian. 
However, it can be seen that setting $\Gamma_1 \coloneq 0$ if it is anti-Hermitian does not change equation \ref{eq:inductive_hypothesis}, since all the anti-Hermitian terms cancel each other out.  

The first step of the algorithm is almost done: however, we observe there is now one new pinnable site, namely site $i$. 
Because site $i$ is pinnable, if we could remove it from the support of $\Gamma_1$ and extract it to $\sigma_1$ somehow, this would permanently remove site $i$ from future $\Gamma_j$ operators. 
After all, $i \notin \text{supp}(M_k)$ for $k \geq 2$.
And indeed, if there is pair of pinnable unpinned sites that are both present in $\Gamma_1$, i.e. if $\Gamma_1 = \Gamma_1^- (i \gamma_i \gamma_k)$ for some Majorana operators $\gamma_i, \gamma_k$ with $k \notin S_1$, then we may write
\begin{equation}
(I+\alpha_1\Gamma_1)\sigma_0 = \frac{1}{2}(I + \alpha_1\Gamma_1^-)(I + i\gamma_i\gamma_k)\sigma_0 + \frac{1}{2}(I - \alpha_1\Gamma_1^-)(I - i\gamma_i\gamma_k)\sigma_0.
\end{equation}

Sampling one of these two terms, setting $\Gamma_1$ to be $\pm \Gamma_1$ and $\sigma_1 = (I \pm i\gamma_i\gamma_k)\sigma_0$ consistently, we have extracted site $i$ and completed one step of the algorithm!
Since eventually every site must be pinned, by the end of the algorithm $\Gamma_{2n-1}$ will have trivial support, i.e. will be a scalar multiple of the identity. 
At this point, $(I + \alpha_{2n -1})\Gamma_{2n-1}$ is a scalar, so defining $\sigma_{2n} \coloneq (I + \alpha_{2n -1})\Gamma_{2n-1}\sigma_{2n-1}$ yields an unnormalized Gaussian state, as desired. 

The biggest impediment to the strategy outlined above is that it must be ensured that the constant term absorbed into $\sigma_{2n-1}$ is nonnegative. 
The algorithm, as currently described, could end with $(I+\alpha_{2n-1}\Gamma_{2n-1})\sigma_{2n-1}$ such that $\alpha_{2n} > 1$, for which it is possible that the scalar $I+\alpha_{2n-1}\Gamma_{2n-1} < 0$, and therefore the operator $\sigma_{2n}$ is not PSD. 
This is an essential detail—it is the reason why the specific bounds on the sampled operators in Lemma \ref{lem:fermion_convexcomb} matter, and precisely why this approach only works for sufficiently high-temperature Gibbs states. 

To remedy this problem, the algorithm must control the growth of $\alpha_j$ in any way possible. 
The choice of the sampling parameters $c_i$ is helpful in preventing the values $\alpha_j$ from growing past 1. 
The algorithm is also altered to make it possible to periodically reset $\Gamma_j$ and $\alpha_j$ to zero. 
Note that whenever $\Gamma_j$ has empty support, $I + \alpha_j \Gamma_j$ is just a scaling factor. 
This factor can therefore be brought into $\sigma_j$ whenever $\Gamma_j$ has empty support, not just in the last step. By redefining $\sigma_j \leftarrow (I+\alpha_j\Gamma_j)\sigma_j$ and $(\Gamma_j, \alpha_j) \leftarrow (0, 0)$, $(I+\alpha_j\Gamma_j)\sigma_j$ is left unchanged.
As long as $\alpha_j$ is at most 1 whenever this step is taken, $\sigma_j$ remains PSD, and $\alpha_j$ is now reset to 0, slowing its growth. 
The algorithm will always choose to pin a site from $S_j$ that is in the support of $\Gamma_j$. This has the effect of driving the support closer to empty.
It is precisely for this reason that the algorithm \textit{must} adaptively choose which site to pin in order to succeed.
A proof of this result is given in Lemma \ref{lem:alpha_bound}.

\subsection{Algorithmic Description}\label{subsec:struc_alg_description}
Now, here is the algorithmic description.
%\jiaqing{I suggest delete the following sentence}\jold{The algorithm takes as input one parameter $N>0$, such that $\beta \leq \frac{1}{2N\mathcal{R}d}$.} 

% Defining $T_j = \text{sup}(\Gamma_j)$, 

% The general approach is to adaptively
% There are several key variables in the algorithm's execution. $\Gamma_j$

% In the algorithm below, we use $\sigma$ to denote the desired probability distribution over fermionic Gaussian states, $\Gamma_j$ to denote the sampled Majorana string (which could be either Hermitian or anti-Hermitian) at $j$-th interation that serves as a building block for constructing the fermionic Gaussian states, $\alpha$ as a scaling factor that dynamically adjusts the contribution of the operator $\Gamma$ to the final output state $\sigma$, and the set $S_j$ that keeps track of the remaining fermionic modes that are yet to be incorporated into the final output $\sigma$ at the $j$-th iteration. 

\begin{algorithm}[Structural Algorithm]
\label{alg:structural_alg}
\mbox{}\\

\textbf{Input:} Fermionic Hamiltonian $H = \sum_a \lambda_a  G_a$ with locality $\mathcal{R}$ and degree $d$; inverse temperature parameter $\beta$ such that $\beta \leq \frac{1}{2N\mathcal{R}}$ for $N>0$. 

\textbf{Output:} An unnormalized Gaussian state $\sigma$.

\begin{enumerate}
    \item Initialize $\sigma_0 = I, \Gamma_0 = 0, \alpha_0 = 0, S_0 = [2n]$. %\jold{[1, \dots, 2n]} 
    \item For $j$ in $0, \dots, 2n-1$: 
        \begin{enumerate}
            \item %\jold{Choose $i$ to be the minimal index in the support of $\Gamma_{j}$; if $\Gamma_{j}$ has an empty support, choose the minimal $i$ in $S_{j}$.}
            If $\text{supp}(\Gamma_j) \cap S_j$ is nonempty, set $i$ as the minimal index in $\text{supp}(\Gamma_j) \cap S_j$. Otherwise, let $i$ be the minimal index in $S_j$.
             Define $S_{j+1} = S_{j} - \{i\}$. Define 
        $$M_{j+1} = e^{\beta H^{S_{j+1}}/2}e^{-\beta H^{S_{j}}/2}.$$
            \item 
          %Sample $I + \beta_1 \Lambda_1$, where $\beta_1$ is of the form $(N\mathcal{R})^{-t_k}$, according to the probability distribution in Lemma \ref{lem:fermion_convexcomb} w.r.t $M_{j+1}$.   
            Sample $I + \beta_1 \Lambda_1$ for $\beta_1 = (N\mathcal{R})^{-t_1}$ from the probability distribution in Lemma \ref{lem:fermion_convexcomb} w.r.t $M_{j+1}$, so that\footnote{For simplicity, with a slight abuse of notation, we denote this sampled $t_k$ as $t_1$ and write $\beta_1 = (N\mathcal{R})^{-t_1}$.}
            \[
            \mathbb{E}\left[I + \beta_1 \Lambda_1\right] = M_{j+1} = e^{\beta H^{S_{j+1}}/2}e^{-\beta H^{S_{j}}/2}.
            \]
            \item  Sample $I + \beta_2 \Lambda_2$ for $\beta_2 = (N\mathcal{R})^{-t_2}$ independently from the same distribution, so that
            \[
            \mathbb{E}\left[I + \beta_2 \Lambda_2\right] = M_{j+1}^\dag.
            \]
            \item Set $\sigma_{j+1} \leftarrow \sigma_{j}$ and $(\Gamma_{j+1}, \alpha_{j+1})$ as:
                \begin{enumerate}
                    \item $\bigl(\Gamma_{j}, (1-\frac{1}{\mathcal{R}})^{-1}\alpha_{j}
                    \bigl)$ with probability $1 - \frac{1}{\mathcal{R}}$
                    \item $\left(\Lambda_1, (6\mathcal{R})\beta_1\right)$ with probability $\frac{1}{6\mathcal{R}}$
                    \item $\left(\Lambda_2, (6\mathcal{R})\beta_2\right)$ with probability $\frac{1}{6\mathcal{R}}$
                    \item $\left(\Lambda_1\Gamma_{j}, (6\mathcal{R})\beta_1\alpha_{j}\right)$ with probability $\frac{1}{6\mathcal{R}}$
                    \item $\left(\Gamma_{j}\Lambda_2, (6\mathcal{R})\alpha_{j} \beta_2\right)$ with probability $\frac{1}{6\mathcal{R}}$
                    \item $\left(\Lambda_1\Lambda_2, (6\mathcal{R})\beta_1 \beta_2\right)$ with probability $\frac{1}{6\mathcal{R}}$
                    \item $\left(\Lambda_1\Gamma_{j}\Lambda_2, (6\mathcal{R})\beta_1\alpha_{j}  \beta_2\right)$ with probability $\frac{1}{6\mathcal{R}}$
                \end{enumerate}
            \item If $\Gamma_{j+1}$ is not Hermitian:
            \begin{itemize}
                \item $\Gamma_{j+1} \leftarrow 0, \alpha_{j+1} \leftarrow 0$.
            \end{itemize}
            \item  If there are two indices $k, l \in S_j^c$ such that $k, l \in \text{supp}\left(\Gamma_{j+1}\right)$, i.e. $\Gamma_{j+1} =\Gamma^-_{j+1} \cdot (i\gamma_k \gamma_l)$ for some operator $\Gamma^-_{j+1}$ whose support is $\supp(\Gamma_{j+1})\backslash\{k,l\}$,: 
            \begin{enumerate}
                \item $\Gamma_{j+1} \leftarrow \Gamma^-_{j+1}$, $\sigma_{j+1} \leftarrow (I + i\gamma_k\gamma_l) \cdot \sigma_{j+1}$ with probability $\frac{1}{2}$
                \item $\Gamma_{j+1} \leftarrow -\Gamma^-_{j+1}$, $\sigma_{j+1} \leftarrow (I - i\gamma_k\gamma_l) \cdot \sigma_{j+1}$ with probability $\frac{1}{2}$
            \end{enumerate}
            \item If $\Gamma_{j+1} \in \{\pm I\}$:
                \begin{itemize}
                    \item $\sigma_{j+1} \leftarrow (I + \alpha_{j+1} \Gamma_{j+1})\sigma_{j+1}$, $\Gamma_{j+1} \leftarrow 0$,  and $\alpha_{j+1} \leftarrow 0$. 
                \end{itemize}
        \end{enumerate}
    \item Set $\sigma \leftarrow \sigma_{2n}$.  
\end{enumerate}

\end{algorithm}

%\jold{The output $\sigma$ of this algorithm, as we will show, is always an unnormalized Gaussian state, and $\mathbb{E}[\sigma] = e^{-\beta H}$.} \jiaqing{I changed the above sentence to a more formal theorem, to make it more convenient to refer to.}

The output of the above algorithm is summarized as follows:
\begin{theorem}\label{thm:structural}
Consider a fermionic Hamiltonian $H = \sum_a \lambda_a  G_a$ with locality $\mathcal{R}$ and degree $d$, and an inverse temperature parameter $\beta$ such that $\beta \leq \frac{1}{2N\mathcal{R}d}$ for $N\geq 24$, the random variable $\sigma$ generated by Algorithm \ref{alg:structural_alg} is always an unnormalized Gaussian state. Moreover, taking into account the randomness in the algorithm, we have
 $$\mathbb{E}[\sigma] = e^{-\beta H}.$$
\end{theorem}
We prove Theorem \ref{thm:structural} in Section \ref{subsec:structural_analysis}.

%we give a remark here. Recall that for any two sets $A,B$,  we use the notation that $A-B = \{a \in A | \,a \notin B\}$, so that the operation is well-defined even when $B\not\subseteq A$. Then in Algorithm \ref{alg:structural_alg} 2(a) it seems possible that $S_{j+1}=S_j$ if $i\not\in S_j$. To avoid confusion, here we emphasize that we will prove that $S_{j+1}\subsetneq S_{j}$ thus $S_{2n}=\emptyset$, that is, Lemma \ref{lem:S2n_empty} in Section \ref{sec:proof_correctness}.}

\subsection{Analysis}\label{subsec:structural_analysis}
\subsubsection{Notation}
First, we establish notation and definitions.
Let $\Gamma_{j+1}^{(x)}$ denotes the value of $\Gamma_{j+1}$ at the end of step $(x)$ of the $j$th iteration of the algorithm, for $x\in \{a,b,...,g\}$. 
Meanwhile, $\Gamma_j$ simply refers to the final value it is set to, i.e. $\Gamma_j^{(g)}$. 
Similar notational conventions are adopted for $\sigma_j$ and $\alpha_j$. 
Moreover, denote $T_j = \text{supp}(\Gamma_j)$ and the pinned sites $P_j = \text{supp}(\sigma_j)$.   
Recall that $S_0=[2n]$ and the sets $S_1,...,S_{2n-1}$ are random variables defined by the Algorithm \ref{alg:structural_alg}.

%\jiaqing{TBC: explain the meaning of filtratioin ``fixing the random bits so far".... }
We will also define a filtration $\{\mathcal{F}_j^{(x)}\}$. Each $\mathcal{F}_j^{(x)}$ is a $\sigma$-algebra that represents all probabilistic events that can be expressed in terms of the randomness in the algorithm until the end of step $(x)$ of the $j$th iteration.
In other words, $\mathcal{F}_j^{(x)}$ represents the information about the outcomes of random choices up until this point. 
In particular, we abbreviate $\mathcal{F}_j = \mathcal{F}_j^{(g)}$, the randomness until iteration $j-1$ ends.
We also extend our definition of $\mathcal{F}_j$ to $j= 0$ by defining $\mathcal{F}_0$ to be the trivial $\sigma$-algebra, since at the beginning of the algorithm no stochastic choices have occurred yet. 
We may then define the conditional expectation $\mathbb{E}\left[X | \mathcal{F}_j^{(x)}\right]$ to be the average of
$X$ over any uncertainty not contained in $\mathcal{F}_j^{(x)}$. 
It may be thought of as the best guess for $X$, conditioned on the knowledge of whether any event in $\mathcal{F}_j^{(x)}$ occurred or did not occur.
For instance, 
\[\mathbb{E}\left[X | \mathcal{F}_{0}\right] = \mathbb{E}\left[X\right],\]
since the trivial $\sigma$-algebra $\mathcal{F}_0$ contains no information, so the best guess is simply the normal expectation. 
The other extreme is that
\[\mathbb{E}\left[X | \mathcal{F}\right] = X\]
whenever the random variable $X$ is perfectly described by the information in $\mathcal{F}$. In this case, $X$ is called $\mathcal{F}$-measurable.  
For instance, if $X$ is defined deterministically by stochastic choices made up to step $x$ of iteration $j$, then $X$ is $\mathcal{F}_j^{(x)}$-measurable. 
One convenient property that we will make use of is that if $X$ is $\mathcal{F}$-measurable, then 
\[\mathbb{E}[XY | \mathcal{F}] = X\mathbb{E}[Y | \mathcal{F}].\]
In other words $\mathbb{E}[\cdot | \mathcal{F}]$ is linear with respect to $\mathcal{F}$-measurable random variables. 

If step $x_1$ of iteration $j_1$ occurs before step $x_2$ of iteration $j_2$, then $\mathcal{F}_{j_1}^{(x_1)} \subseteq \mathcal{F}_{j_2}^{(x_2)} $. In general, for any two $\sigma$-algebras if $\mathcal{G}_1 \subseteq \mathcal{G}_2$, then for some random variable $X$,
\[\mathbb{E}\left[\mathbb{E}\left[X | \mathcal{G}_2\right] | \mathcal{G}_1\right]= \mathbb{E}\left[X| \mathcal{G}_1\right].\]
This is the law of total expectation. In our setting, it expresses that averaging over the randomness beginning after step $(j_2, x_2)$, then averaging over the randomness beginning after $(j_1, x_1)$, is equivalent to simply averaging over the randomness beginning after $(j_1, x_1)$. 
Each of the above properties will prove useful in the rigorous probabilistic analysis of the algorithm. 

\subsubsection{Proof of Correctness}\label{sec:proof_correctness}

To aid our proof,  
%\jold{we will first show a lemma with two statements} 
we will show two statements concerning the relationship between the sets $S_j, T_j,$ and $P_j$ in the following Lemma \ref{lem:supp_analysis}.
Firstly, the pinned sites, $P_j$, are disjoint from $S_j \cup T_j$, which are the remaining sites that could be in the support of any $M_i$ for $i > j$ or in the support of $\Gamma_j$. 
Secondly, $|T_j - S_j| \leq 1$, which has the interpretation that the support of $\Gamma_j$ can only include one element outside of $S_j$. In other words, there is only one site $i$ that is pinnable, i.e. $i \in S_j^c$, but unpinned, i.e. $i \notin P_j$. 
Here, we are using the notational convention that if $A$ and $B$ are sets, then $A-B = \{a \in A | \,a \notin B\}$, so that the operation is well-defined even when $B\not\subseteq A$.

Lemma \ref{lem:supp_analysis} indicates that as $S_j$ becomes smaller and smaller, the support of $\Gamma_j$ eventually dwindles as well.

\begin{lemma}\label{lem:supp_analysis}
For any $j\in 0, \dots, 2n$, $(S_j \cup T_j) \cap P_j = \emptyset$, and $|T_j - S_j |\leq 1$.
\end{lemma}
\begin{proof}
We can prove $(S_j \cup T_j) \cap P_j = \emptyset$ by induction. 
Firstly, for $j = 0$, $P_0 = \emptyset$ since $\sigma_0 = I$. 

Now, note that 
\begin{align}
    S_{j+1} \cup T_{j+1} \subseteq S_j \cup T_{j+1} \subseteq S_j \cup T_j.\label{eq:ST}
\end{align}
The first inclusion is true because $S_{j+1} \subseteq S_j$. For the second inclusion, the support of $\Gamma_{j+1}$ is contained in the support of $\Gamma_{j+1}^{(d)}$, which is some product of $\Lambda_1, \Lambda_2, $ and $\Gamma_j$, and therefore the support is contained in the union of $\text{supp}(\Gamma_j)$ (which is $T_j$ by definition), $\text{supp}(\Lambda_1) \subseteq S_{j}$, and $\text{supp}(\Lambda_2) \subseteq S_{j}$, where the last two inclusions come from Lemma \ref{lem:fermion_convexcomb}. 
Therefore, we also show that
\begin{align}
    T_{j+1} \subseteq \text{supp}\left(\Gamma_{j+1}^{(d)}\right) \subseteq S_j \cup T_j,\label{eq:set_T}
\end{align}
% \jiaqing{in the following proof we have cited this equation, so I labeled it.}
so that $S_{j}\cup T_{j+1} \subseteq S_{j} \cup T_j$. 

Decomposing $P_{j+1} = P_{j} \cup (P_{j+1} - P_{j})$, we may rewrite 
\begin{align*}
(S_{j+1} \cup T_{j+1}) \cap P_{j+1} &\subseteq \left((S_{j} \cup T_{j+1}) \cap P_j\right) \cup \left((S_{j} \cup T_{j+1}) \cap (P_{j+1} - P_j)\right)  \\ &\subseteq \left((S_{j} \cup T_{j}) \cap P_j\right) \cup \left((S_{j} \cup T_{j+1}) \cap (P_{j+1} - P_j)\right) \\ &= (S_{j} \cup T_{j+1}) \cap (P_{j+1} - P_j)
\end{align*}
where the second $\subseteq$ follows from Eq.~(\ref{eq:ST}), and the last equality holds by the inductive hypothesis. 
Furthermore, note that the elements that can be added from $P_j$ to $P_{j+1}$ (in step $(f)$) are indices $k, l$ such that $k, l \notin S_j$, and by construction $k, l \notin T_{j+1}$, since $\gamma_k$ and $\gamma_l$ are removed from the support of $\Gamma_{j+1}$. 
It follows that $(S_{j} \cup T_{j+1}) \cap (P_{j+1} - P_j) = \emptyset$, which completes the inductive hypothesis. 

Now, we would like to prove that $|T_j - S_j| \leq 1$, again by induction. 
Since $S_0 = [2n]$, the base case follows immediately. 
Now, we assume the induction hypothesis. 
In Eq.~(\ref{eq:set_T}) we showed that $T_{j+1} \subseteq \text{supp}\left(\Gamma_{j+1}^{(d)}\right) \subseteq S_{j} \cup T_j$, so $T_{j+1} - S_j \subseteq \text{supp}\left(\Gamma_{j+1}^{(d)}\right) - S_j \subseteq T_{j} - S_j$. 
It follows that
\begin{align}
T_{j+1} - S_{j+1} &\subseteq (T_{j+1} - S_j) \cup (S_{j} - S_{j+1})\nonumber \\  &\subseteq (T_j  - S_j) \cup (S_{j} - S_{j+1}) \label{eq:TSj1}
\end{align}
%\jnew{According to Algorithm \ref{alg:structural_alg} 2(a), the set $S_j - S_{j+1}$ consists of at most one element. Without loss of generality we consider the case when $S_j - S_{j+1}$ is non-empty, say $S_j - S_{j+1} = \{i\}$, otherwise  by Eq.~(\ref{eq:TSj1}) and the inductive hypothesis we conclude that $|T_{j+1} - S_{j+1}|= |T_{j} - S_{j}|\leq 1$.}
The set $S_j - S_{j+1}$ consists of one element,
say $S_j - S_{j+1} = \{i\}$. 
If $|T_j - S_j| = 0$, then $T_{j+1} - S_{j+1} \subseteq S_{j} - S_{j+1}=\{i\}$, so $|T_{j+1} - S_{j+1}| \leq 1$. On the other hand, if $|T_j - S_j| = 1$, then potentially the largest value for $|T_{j+1} - S_{j+1}|$ is two.
However, in this case, by $T_{j+1} \subseteq \text{supp}(\Gamma_{j+1}^{(d)})$, it would also be true that $\left|\text{supp}\left(\Gamma_{j+1}^{(d)}\right) - S_{j+1}\right| \geq 2$, and therefore the cardinality is precisely 2 since $\text{supp}\left(\Gamma_{j+1}^{(d)}\right) - S_{j+1} \subseteq (S_j \cup T_j) - S_{j+1} \subseteq (T_j - S_{j})\cup\{i\}$. This implies that there are precisely two indices in the support of $\Gamma_{j+1}^{(d)}$ that are in $S_{j+1}^{c}$. 
However, these two indices would then have been pinned immediately in step $(f)$, so that $\left|T_{j+1} - S_{j+1}\right| =0$. 
This completes the proof. 
\end{proof}

Several corollaries follow easily from this lemma.

\begin{corollary}\label{cor:Gaussianstate}
Each $\sigma_j$, and in particular the output $\sigma = \sigma_{2n}$ of the algorithm, is an unnormalized Gaussian state. 
\end{corollary}
\begin{proof}
Whenever step $(f)$ is run and a term of the form $I \pm i\gamma_k\gamma_l$ is factored into $\sigma_{j+1}$, we have that $k, l \in \text{supp}\left(\Gamma_{j+1}^{(d)}\right) \subseteq T_j \cup S_j$, implying that $k, l \notin P_j$ by Lemma \ref{lem:supp_analysis}. 
We now use the fact that if $M$ is some perfect matching of the Majorana operators for a system, then $\frac{1}{\text{tr}(I)}\prod_{(k, l) \in M} (I \pm i\gamma_k \gamma_l)$ is a Gaussian state~\cite{herasymenko2023optimizing}. Thus, each new fermionic mode added to $\sigma_j$ adds pairs of disjoint Majorana operators, and therefore each $\sigma_{j}$ will be an unnormalized Gaussian state, with a scaling factor obtained by various iterations of step $(g)$.
\end{proof}

The above corollary importantly does not imply that the normalization factors of these Gaussian states are nonnegative. 

\begin{corollary}\label{cor:final_gamma}
The operator $\Gamma_{2n}=0$, so that $(I+\alpha_{2n}\Gamma_{2n})\sigma_{2n} = \sigma_{2n}$. 
\end{corollary}
\begin{proof}
Note that $|T_{2n}|=|T_{2n} - S_{2n}| \leq 1$ by Lemma \ref{lem:supp_analysis}, and the support of $\Gamma_{2n}$ is even (it is contained in $\mathcal{M}^*$), so $T_{2n} = \emptyset$. 
It follows that $\Gamma_{2n} \in \{\pm I, 0\}$, and therefore $\Gamma_{2n} = 0$ by step (g).
\end{proof}

\begin{corollary}\label{cor:hermitian_output}
Each operator $(I + \alpha_j\Gamma_j)\sigma_j$ is Hermitian. 
\end{corollary}
\begin{proof}
Since $T_j \cap P_j = \emptyset$, $\sigma_j$ and $I+\alpha_j\Gamma_j$ commute.
By Corollary \ref{cor:Gaussianstate}, $\sigma_j$ is Hermitian, $\alpha_j$ is real, and each $\Gamma_j$ is Hermitian, so the corollary follows. 
\end{proof}

\begin{lemma}\label{lem:fermion_expectation}
The output distribution of the algorithm has expectation $e^{-\beta H}$. 
\end{lemma}
\begin{proof}
In order to prove this result, we will show that 
\begin{equation}\label{eq:inductiveexpectation}\mathbb{E}[(I+\alpha_{j+1}\Gamma_{j+1})\sigma_{j+1} | \mathcal{F}_{j}] = M_{j+1}(I+\alpha_{j}\Gamma_{j})\sigma_{j} M_{j+1}^\dag.
\end{equation}
This result can be rewritten as 
\begin{equation}\label{eq:rewritteninductiveexpectation}\mathbb{E}[M_{j+1}^{-1}(I+\alpha_{j+1}\Gamma_{j+1})\sigma_{j+1}(M_{j+1}^\dag)^{-1} | \mathcal{F}_{j}] = (I+\alpha_{j}\Gamma_{j})\sigma_{j},\end{equation}
since $M_{j+1}$ is defined deterministically in terms of $\Gamma_j$ and $S_j$, and so $M_{j+1}$ is $\mathcal{F}_j$-measurable. 
The lemma will follow from this equation, as follows. We may write:
\begin{align*}
&\mathbb{E}\left[M_1^{-1}\dots M_{2n}^{-1}(I+\alpha_{2n}\Gamma_{2n})\sigma_{2n}(M_{2n}^{\dag})^{-1} \dots (M_{1}^{\dag})^{-1}\right] \\
&= \mathbb{E}\left[\mathbb{E}\left[M_1^{-1}\dots M_{2n-1}^{-1}M_{2n}^{-1}(I+\alpha_{2n}\Gamma_{2n})\sigma_{2n}(M_{2n}^{\dag})^{-1} (M_{2n-1}^{\dagger})^{-1}\dots (M_{1}^{\dag})^{-1}| \mathcal{F}_{2n-1}\right]\right]
\\
&= \mathbb{E}\left[M_1^{-1}\dots M_{2n-1}^{-1}\mathbb{E}\left[M_{2n}^{-1}(I+\alpha_{2n}\Gamma_{2n})\sigma_{2n}(M_{2n}^{\dag})^{-1} | \mathcal{F}_{2n-1}\right](M_{2n-1}^{\dagger})^{-1}\dots (M_{1}^{\dag})^{-1}\right]\\
&= \mathbb{E}\left[M_1^{-1}\dots (I+\alpha_{2n-1}\Gamma_{2n-1})\sigma_{2n-1}\dots (M_{1}^{\dag})^{-1}\right]
\\
&= \dots = \mathbb{E}\left[I + \alpha_0 \Gamma_0\right]\\\
&= I. 
\end{align*}
Here, we use the law of total expectation in the first equality, use that $M_1, \dots, M_{2n-1}$ are $\mathcal{F}_{2n-1}$-measurable in the second equality, and apply equation $\ref{eq:rewritteninductiveexpectation}$ in the third equality. Repeating this $2n$ times for each $\mathcal{F}_j$ until $\mathcal{F}_0$ yields the final expression. 
Since $M_1^{-1}\dots M_{2n}^{-1} = e^{\beta H/2}$, the above formula implies that
\[\mathbb{E}[e^{\beta H/2}(I+\alpha_{2n}\Gamma_{2n})\sigma_{2n}e^{\beta H/2}] = I,\]
and since $e^{\beta H/2}$ is a constant we conclude that $\sigma_{2n} = (I+\alpha_{2n}\Gamma_{2n})\sigma_{2n}= e^{-\beta H}$. 

To obtain equation \ref{eq:inductiveexpectation}, we first calculate the conditional expectation of $\left(I+\Gamma_{j+1}^{(d)}\right)\sigma_{j+1}^{(d)}$ with respect to $\mathcal{F}_{j+1}^{(c)}$:
\begin{equation}
\begin{aligned}
\mathbb{E}\left[\left(I+\Gamma_{j+1}^{(d)}\right)\sigma_{j+1}^{(d)} |\mathcal{F}_{j+1}^{(c)}\right] &= \Bigg(\left(1 - \frac{1}{\mathcal{R}}\right)\left(I + \left(1 - \frac{1}{\mathcal{R}}\right)^{-1}\alpha_j \Gamma_j \right) \\ &\hspace{0.5 cm}+ \frac{1}{6\mathcal{R}}(I + 6\mathcal{R}\beta_1 \Lambda_1)\\  &\hspace{0.5 cm}+ \frac{1}{6\mathcal{R}}(I + 6\mathcal{R}\beta_2 \Lambda_2)
\\  &\hspace{0.5 cm}+ \frac{1}{6\mathcal{R}}(I + 6\mathcal{R}\beta_1\alpha_j\Lambda_1\Gamma_j)
\\  &\hspace{0.5 cm}+ \frac{1}{6\mathcal{R}}(I + 6\mathcal{R}\alpha_j\beta_2 \Gamma_j\Lambda_2)
\\  &\hspace{0.5 cm}+ \frac{1}{6\mathcal{R}}(I + 6\mathcal{R}\beta_1\beta_2 \Lambda_1\Lambda_2)
\\  &\hspace{0.5 cm}+ \frac{1}{6\mathcal{R}}(I + 6\mathcal{R}\beta_1\alpha_j\beta_2 \Lambda_1 \Gamma\Lambda_2)\Bigg)\sigma_j.
\end{aligned}
\end{equation}
Grouping together the $I$ terms and simplifying, we obtain the inner term as
\[I + \alpha_j \Gamma_j + \beta_1 \Lambda_1 + \beta_2 \Lambda_2 + (\beta_1 \Lambda_1)(\alpha_j \Gamma_j)+(\alpha_j \Gamma_j)(\beta_2 \Lambda_2)+(\beta_1 \Lambda_1)(\beta_2 \Lambda_2)+(\beta_1 \Lambda_1)(\alpha_j \Gamma_j)(\beta_2 \Lambda_2).\]
Including $\sigma_j$ in the expression yields 
\[(I+ \beta_1\Lambda_1)(I+\alpha_j \Gamma_j)(I + \beta_2\Lambda_2)\sigma_j.\]
Using the law of total expectation, we may now take the expectation of $\left(I+\alpha^{(d)}_{j+1} \Gamma_{j+1}^{(d)}\right)\sigma_{j+1}^{(d)}$ with respect to $\mathcal{F}_{j}$, obtaining
\begin{equation}
\begin{aligned}
\mathbb{E}\left[\left(I+\alpha^{(d)}_{j+1} \Gamma_{j+1}^{(d)}\right)\sigma_{j+1}^{(d)} | \mathcal{F}_{j}\right] &= \mathbb{E}[(I+ \beta_1\Lambda_1)](I+\alpha_j \Gamma_j)\mathbb{E}[(I + \beta_2\Lambda_2)]\sigma_j \\&=M_{j+1}(I+\alpha_j \Gamma_j)M_{j+1}^\dag\sigma_j \\&=  M_{j+1}(I+\alpha_j \Gamma_j)\sigma_jM_{j+1}^\dag
\end{aligned}
\end{equation}
where the last equality holds since $S_j \cap P_j = \emptyset$, so $\sigma_j$ commutes with $M_{j+1}$.

In step (e), every non-Hermitian Majorana operator is set to 0. 
However, the operators that are not Hermitian are anti-Hermitian. Using this fact, we can expand
\begin{equation}
\begin{aligned}
 \left(I+\alpha^{(e)}_{j+1} \Gamma_{j+1}^{(e)} \right)\sigma_{j+1}^{(e)} &=  \left(I+\frac{1}{2}\alpha_{j+1}^{(d)}\left(\Gamma_{j+1}^{(d)} + \left(\Gamma_{j+1}^{(d)}\right)^\dag\right)\right)\sigma_{j+1}^{(d)}
 \\&=\frac{1}{2}\left(\left(I+\alpha^{(d)}_{j+1} \Gamma_{j+1}^{(d)}\right)\sigma_{j+1}^{(d)} + \left(\left(I+\alpha^{(d)}_{j+1} \Gamma_{j+1}^{(d)}\right)\sigma_{j+1}^{(d)}\right)^\dag\right)
\end{aligned}
\end{equation}
The first equality holds because $\frac{1}{2}(\Gamma+\Gamma^\dag) = \Gamma$ if $\Gamma$ is Hermitian and 0 if $\Gamma$ is anti-Hermitian. The second equality is satisfied because $\sigma_{j+1}^{(d)}$ is Hermitian, and commutes with $\Gamma_{j+1}^{(d)}$ since $\text{supp}\left(\Gamma_{j+1}^{(d)}\right) \subseteq S_j \cup T_j$, so it has disjoint support from $\sigma_{j+1}^{(d)} = \sigma_j$ by Lemma \ref{lem:supp_analysis}. 
Now, taking the conditional expectation with respect to $\mathcal{F}_{j}$, we find
\begin{equation}
\begin{aligned}
 &\mathbb{E}\left[\left(I+\alpha^{(e)}_{j+1} \Gamma_{j}^{(e)} \right)\sigma_{j+1}^{(e)} | \mathcal{F}_{j}\right]\\ &=   \frac{1}{2}\mathbb{E}\left[\left(I+\alpha^{(d)}_{j+1} \Gamma_{j+1}^{(d)} \right)\sigma_{j+1}^{(d)} | \mathcal{F}_{j}\right] +  \frac{1}{2}\mathbb{E}\left[\left(I+\alpha^{(d)}_{j+1} \Gamma_{j+1}^{(d)} \right)\sigma_{j+1}^{(d)} | \mathcal{F}_{j}\right]^\dag\\
 &= \frac{1}{2}M_{j+1}(I+\alpha_j\Gamma_j)\sigma_j M_{j+1}^\dag + \frac{1}{2}\left(M_{j+1}(I+\alpha_j\Gamma_j)\sigma_j M_{j+1}^\dag\right)^\dag
 \\
 &= M_{j+1}(I+\alpha_j\Gamma_j)\sigma_j M_{j+1}^\dag,
\end{aligned}
\end{equation}
where the last equality holds since $(I+\alpha_j\Gamma_j)\sigma_j$ is Hermitian by Corollary \ref{cor:hermitian_output}. 

Now, consider $I+\alpha^{(f)}_{j+1} \Gamma_{j+1}^{(f)}$. Here, we may write
\begin{equation}
\begin{aligned}
&\mathbb{E}\left[\left(I+\alpha^{(f)}_{j+1} \Gamma_{j+1}^{(f)}\right)\sigma_{j+1}^{(f)}| \mathcal{F}_{j+1}^{(e)}\right]\\ &= \frac{1}{2}\left(I+\alpha^{(e)}_{j+1} \Gamma^-_{j+1}\right)(I+i\gamma_k\gamma_l)\sigma_{j+1}^{(e)} + \frac{1}{2}\left(I-\alpha_{j+1}^{(e)} \Gamma^-_{j+1}\right)(I-i\gamma_k\gamma_l)\sigma_{j+1}^{(e)} \\&= (I + \alpha_{j+1}^{(e)} \Gamma^-_{j+1}(i\gamma_k\gamma_l))\sigma_{j+1}^{(e)} \\
&= \left(I + \alpha_{j+1}^{(e)} \Gamma_{j+1}^{(e)}\right)\sigma_{j+1}^{(e)}.
\end{aligned}
\end{equation}
Finally,
\begin{equation}
\begin{aligned}
\left(I+\alpha_{j+1} \Gamma_{j+1}\right)\sigma_{j+1} &= \left(I+0\right)\left(I+\alpha^{(f)}_{j+1} \Gamma_{j+1}^{(f)}\right)\sigma_{j+1}^{(f)}\\
&= \left(I+\alpha^{(f)}_{j+1} \Gamma_{j+1}^{(f)}\right)\sigma_{j+1}^{(f)}
\end{aligned}
\end{equation}
Chaining these equalities and using the law of total expectation, we obtain the desired formula, 
\begin{equation}
\begin{aligned}
\mathbb{E}\left[\left(I+\alpha_{j+1} \Gamma_{j+1}\right)\sigma_{j+1} | \mathcal{F}_{j}\right] &= \mathbb{E}\left[\left(I+\alpha_{j+1}^{(e)} \Gamma_{j+1}^{(e)}\right)\sigma_{j+1}^{(e)} | \mathcal{F}_{j}\right]\\
&= M_{j+1}(I+\alpha_j\Gamma_j)\sigma_j M_{j+1}^\dag.
\end{aligned}
\end{equation}
\end{proof}

\begin{lemma}\label{lem:alpha_bound}
Say that the parameter $N$ of the structural algorithm satisfies $N \geq 2$. Then, for any $j$, $\alpha_j \leq \frac{24}{N}\left(1-\frac{1}{\mathcal{R}}\right)^{|T_j| - b_j}$, where $b_j = |T_j - S_j|$.  
\end{lemma}
\begin{proof}
In the base case $\alpha_0 = 0$, which clearly satisfies the hypothesis. For the inductive step, assume that $\alpha_j \leq (1-\frac{1}{\mathcal{R}})^{|T_j| - b_j}$.
We may assume that step $(e)$ and step $(g)$ are not taken, since they set $\alpha_{j+1}$ to zero, in which case the inductive step is immediate. 

The value of $\alpha_{j+1}$ is therefore decided in step (d). 
If the update (i) in (d) is taken, then 
\begin{equation}
\alpha_{j+1}=\left (1-\frac{1}{\mathcal{R}}\right)^{-1}\alpha_{j} \leq \frac{24}{N}\left(1-\frac{1}{\mathcal{R}}\right)^{|T_j| - b_j-1}.
\end{equation}
If $b_j = 1$ at the beginning of the step, then $|T_{j+1}| = |\Gamma^-_{j+1}| \leq \left|\Gamma^{(d)}_{j+1}\right|-2  = |\Gamma_{j}| - 2$ since $\Gamma^{(d)}_{j+1}$ has two Majorana operators with indices $i, k \in S_{j+1}^c$ that are factored out together. These operators are factored out precisely because $\Gamma_{j+1}^{(d)} = \Gamma_{j}$ when update $(i)$ is taken, so they have the same support. 
Since $b_j=1$ there was already one index $k \in T_j - S_j$. 
Moreover, $T_j$ always has even cardinality, so there must exist some index (taken to be minimal) $i \in T_j \cap S_j^c$.
Since $S_{j+1}=S_j - \{i\}$ and $i$ is in the support of $\Gamma_{j+1}^{(d)} = \Gamma_{j}$, we have $i, k \in \text{supp}\left(\Gamma_{j+1}^{(d)}\right) - S_{j+1}$ and so both are factored out in step $(f)$.
Now $b_{j+1} = 0$, since the only two elements in $\text{supp}\left(\Gamma_{j+1}^{(d)}\right)- S_{j+1}$ were removed from the support of $\Gamma_{j}$ to produce $\Gamma_{j+1}$. It follows that $|T_{j+1}| - b_{j+1} \leq |T_j| - b_j-1$, so
\[\alpha_{j+1} \leq \frac{24}{N}\left(1 - \frac{1}{\mathcal{R}}\right)^{|T_{j+1}| - b_{j+1}}\]
as desired. 
If $\Gamma$ already has empty support, then $\alpha$ is set to 0 in step (g). 
On the other hand, if $b_j = 0$, then since $i$ is removed from $S_j$ in step (a) and $\Gamma_{j+1} = \Gamma_j$, $b_{j+1} = 1$.
It follows that $\alpha_{j+1} \leq \frac{24}{N}\left(1 - \frac{1}{\mathcal{R}}\right)^{|T_{j+1}| - b_{j+1}}$ as desired. 

We now consider the remaining possible update steps. 
For (ii), (iii), note that $\frac{6\mathcal{R}}{N \mathcal{R}} = \frac{24}{N}\cdot \frac{1}{4} \leq \frac{24}{N}(1-\frac{1}{\mathcal{R}})^\mathcal{R}$ for $\mathcal{R} \geq 2$, as we have assumed. Moreover, since $N \geq 2$, $(N \mathcal{R})^{-k} \leq 4^{-k} \leq (1-\frac{1}{\mathcal{R}})^{\mathcal{R} k}$. 
Together, these statements demonstrate that $(6\mathcal{R})(N\mathcal{R})^{-k} \leq \frac{24}{N}\left(1 - \frac{1}{\mathcal{R}}\right)^{\mathcal{R}k}$, and so in cases (ii) and (iii) we may deduce that
\begin{equation}
\begin{aligned}
\alpha_{j+1} = 6 \mathcal{R}(N\mathcal{R})^{-t_1} &\leq \frac{24}{N}\left(1-\frac{1}{\mathcal{R}}\right)^{\mathcal{R}t_1}\leq \frac{24}{N}\left(1-\frac{1}{\mathcal{R}}\right)^{|\Lambda_1|},\\
\alpha_{j+1} = 6 \mathcal{R}(N\mathcal{R})^{-t_2} &\leq \frac{24}{N}\left(1-\frac{1}{\mathcal{R}}\right)^{\mathcal{R}t_2}\leq \frac{24}{N}\left(1-\frac{1}{\mathcal{R}}\right)^{|\Lambda_2|}.
\end{aligned}
\end{equation}
since $|\Lambda_1| \leq \mathcal{R}t_1$ and $|\Lambda_2| \leq \mathcal{R}t_2$ by Lemma \ref{lem:fermion_convexcomb}. Since in cases (ii) and (iii), $\Gamma_{j+1}$ has support contained in $\text{supp}(\Lambda_1)$ and $\text{supp}(\Lambda_2)$ respectively, we conclude
\begin{align}
\alpha_{j+1} \leq \frac{24}{N}\left(1 - \frac{1}{\mathcal{R}}\right)^{|T_{j+1}| }\leq \frac{24}{N}\left(1 - \frac{1}{\mathcal{R}}\right)^{|T_{j+1}| - b_{j+1}}.
\end{align}
Similarly, for case (vi), using the inequality
\begin{equation}
\begin{aligned}
\alpha_{j+1} &= 6 \mathcal{R}(N\mathcal{R})^{-(t_1+t_2)} \\&\leq \frac{24}{N}\left(1-\frac{1}{\mathcal{R}}\right)^{\mathcal{R}(t_1+t_2)}\\&\leq \frac{24}{N}\left(1-\frac{1}{\mathcal{R}}\right)^{|\Lambda_1\Lambda_2|} \leq \frac{24}{N}\left(1 - \frac{1}{\mathcal{R}}\right)^{|T_{j+1}| - b_{j+1}}
\end{aligned}
\end{equation}

Cases (iv), (v), and (vii) remain. These three cases are dealt with very similarly, so we only reason through case (vii). 

There are two possibilities worth considering. In the first, $b_{j+1} \geq b_j$. Here, we write that
\begin{equation}
\begin{aligned}
\alpha_{j+1} &\leq 6\mathcal{R} (N\mathcal{R})^{-(t_1+t_2)}\left(1 - \frac{1}{\mathcal{R}}\right)^{|T_{j}|- b_j}\\&\leq \frac{24}{N}\left(1 - \frac{1}{\mathcal{R}}\right)^{|T_{j}|- b_j + \mathcal{R}(t_1+t_2)}\leq \frac{24}{N}\left(1 - \frac{1}{\mathcal{R}}\right)^{|T_{j+1}|- b_{j+1} },
\end{aligned}
\end{equation}
since $|T_{j+1}| = |\Gamma_{j+1}| \leq |\Lambda_1\Gamma_j\Lambda_2| \leq |T_j| + \mathcal{R}(t_1+t_2)$. 

Now, say that $b_{j+1} < b_j$, i.e. $b_{j+1} = 0$ and $b_j = 1$. If we say $\{i\} = S_{j}- S_{j+1}$, then we know that $i \in S_{j+1}^{c}$, and therefore $i \notin \Gamma_{j+1}$, or else it would be true that $i \in T_{j+1}- S_{j+1} $ and $b_{j+1} > 0$.
We may therefore remark that 
\begin{equation}
|\Gamma_{j+1}| \leq |\Lambda_1| +  |\Gamma_j| + |\Lambda_2| - 1\leq |T_j| + \mathcal{R}(t_1 + t_2) - 1.
\end{equation}
The first inequality arises precisely because $\Gamma_{j+1}$ has support contained in the union of $\Lambda_1$, $\Gamma_j$, and $\Lambda_2$, but also does not contain the index $i \in T_j$. 

Just as for the first possibility, it now follows that 
\begin{align*}
\alpha_{j+1} &\leq 6\mathcal{R} (N\mathcal{R})^{-(t_1+t_2)}\left(1 - \frac{1}{\mathcal{R}}\right)^{|T_{j}|- b_j}\\&\leq \frac{24}{N}\left(1 - \frac{1}{\mathcal{R}}\right)^{|T_{j}|- b_j - 1+ \mathcal{R}(t_1+t_2)}\leq \frac{24}{N}\left(1 - \frac{1}{\mathcal{R}}\right)^{|T_{j+1}|- b_{j+1} }.
\end{align*}

We may conclude the inductive hypothesis no matter which choice is taken in step (d), completing the proof.
\end{proof}

\begin{lemma}\label{lem:fermion_nonnegative_trace}
For any $\beta \leq \frac{1}{48\mathcal{R}d}$, and choosing $N = 24$ in the structural algorithm, the trace of any output state $\sigma$ of the algorithm satisfies $\text{tr}(\sigma) \geq 0$. 
\end{lemma}
\begin{proof}
Since $\beta \leq \frac{1}{48\mathcal{R}d} = \frac{1}{2N\mathcal{R}d}$, the previous lemma applies. 
To prove this result, it suffices to prove that whenever step $(g)$ of the algorithm is undertaken, $\alpha_j \leq 1$, so that the scaling factor $I + \alpha_j \Gamma_j$ is nonnegative. 
Indeed, other than these factors, $\sigma_{2n}$ can be expressed as $\Pi_{(k, l) \in M} (I \pm i\gamma_k\gamma_l)$. 
Expanding out this product, since the pairs $(k, l)$ do not overlap, every summand has nonzero support except $I$, and so 

\[\text{tr}\left(\Pi_{(k, l) \in M} (I \pm i\gamma_k\gamma_l)\right) = \text{tr}(I) \geq 0.\]
If each of the step $(g)$ factors is nonnegative, we obtain the desired property $\text{tr}(\sigma_{2n}) \geq 0$. Now, we can simply apply Lemma \ref{lem:alpha_bound}. When time step $(g)$ is taken, $|T_j| = |\Gamma_j| = 0$ and therefore $b_j = 0$ as well. So the lemma implies that 
\begin{equation}
\alpha_j \leq \frac{24}{N}\left(1-\frac{1}{\mathcal{R}}\right)^{|T_j| - b_j} = \frac{24}{N} = 1.
\end{equation}
\end{proof}

\begin{corollary}
For $\beta \leq \frac{1}{48\mathcal{R}d}$, The Gibbs state $\frac{e^{-\beta H}}{\text{tr}(e^{-\beta H})}$ is a convex combination of Gaussian states. 
\end{corollary}
\begin{proof}
By Corollary \ref{cor:Gaussianstate}, Lemma \ref{lem:fermion_expectation}, and Lemma \ref{lem:fermion_nonnegative_trace}, $e^{-\beta H}$ is a convex combination of unnormalized Gaussian states. 
This implies that the Gibbs state is a convex combination of Gaussian states. 
\end{proof}

\section{Sampling Algorithm}\label{sec:sampling_alg}
\subsection{Overview}
In this section, we upgrade the structural algorithm in Section \ref{sec:structure_alg} to a sampling algorithm that efficiently samples from the Gibbs state of a $\mathcal{R}$-local, degree-$d$ fermionic Hamiltonian at sufficiently high temperature. 
Specifically, when the inverse temperature satisfies \[\beta \leq \frac{1}{50\mathcal{R}d^2},\] we show that a classical algorithm can efficiently sample from a distribution of fermionic Gaussian states, such that the expectation approximates the Gibbs state \[\rho_\beta  = \frac{e^{-\beta H}}{\text{tr}(e^{-\beta H})}\]
up to trace distance $\epsilon$.
This extends the structural result in Section \ref{sec:structure_alg} and shows that Gibbs sampling is classically easy in the high-temperature fermionic regime as well. 

We will first comment on the reduced temperature threshold needed for the sampling algorithm, detailed in Section \ref{subsec:temp_threshold}. 
First, we slightly tighten the temperature bound compared to Section \ref{sec:structure_alg}, from $\frac{1}{48\mathcal{R}^2}$ to $\frac{1}{50\mathcal{R}^2}$, to ensure that the coefficients $\alpha$'s arising in the telescoping expansion is bounded away from $1$. 
Previously, only $\alpha \leq 1$ was needed, to ensure that the intermediate operators in the form of $I + \alpha \Gamma$ remained positive semidefinite.
In Lemma \ref{lem:alpha_bound_sampling} below, this bound is amplified slightly.
Additionally, we now need to efficiently approximate the partition function $\text{tr}(e^{-\beta H})$, which is necessary to normalize samples.
This introduces another factor of $d^{-1}$ needed for the temperature threshold, as shown in Corollary \ref{cor:partition_estimate}.

In order to design an efficient sampling algorithm, we reinterpret the structural algorithm from Section $\ref{sec:structure_alg}$ as implicitly generating a sampling tree $\mathcal{T}$, introduced in Section \ref{subsec:sample_tree}.
Each node in $\mathcal{T}$ represents a step of the sampling process, and each edge corresponds to a stochastic outcome of that step, arising from the random choices of sampling an operator from the expansion of $M$ or $M^{\dagger}$, or performing an update or pinning step. 
A complete path from the root to a leaf in $\mathcal{T}$ specifies a sequence of random choices leading to an unnormalized fermionic Gaussian state, and each path carries a weight determined by the product of probabilities along the path. 
The distribution over all leaves encodes a convex decomposition of $e^{-\beta H}.$
However, this is a distribution of unnormalized Gaussian states with varying normalization constants. 
The algorithm that samples from $\mathcal{T}$ with a distribution that must be corrected by the normalization constants. 
$\mathcal{T}$ is also not efficient to sample from, for instance because nodes can have infinitely many children. 
We show that by simulating a well-chosen random walk (Theorem \ref{thm:simulate_distr}) on an approximate version of $\mathcal{T}$, $\mathcal{T}'$, an efficient sampling algorithm for the corrected distribution may be obtained.
$\mathcal{T}'$ is defined in Section \ref{subsec:truncated_sample_tree}. 
Analyzing the performance of this random walk, we show that high-temperature fermionic Gibbs states can be efficiently sampled. 

% In Section \ref{subsec:sample_tree}, we introduce a sample tree $\mathcal{T}$, which represents all possible execution paths of the structural algorithm. 
% Each path in $\mathcal{T}$ corresponds to a particular sequence of random sampling steps and produces an unnormalized fermionic Gaussian state.
% The weight of each node is proportional to the product of probabilities encountered along that path.
% While the full sample tree $\mathcal{T}$ encodes the correct distribution over unnormalized Gaussian operators whose expectation gives $e^{-\beta H}$, it is computationally inefficient to sample from it due to the exponential number of paths. 
% To overcome this, in Section \ref{subsec:truncated_sample_tree}, we introduce a truncated sample tree $\mathcal{T}'$, which prune away branches of $\mathcal{T}$ that 

\subsection{Temperature Threshold} \label{subsec:temp_threshold}
% In this section, the goal is to use the structural algorithm described in Section \ref{sec:structure_alg} as a subroutine to prove that Gibbs sampling is classically easy when $\beta \leq \frac{1}{50\mathcal{R}d^2}$. 
In the sampling setting, the bound on the inverse temperature $\beta$ changes from the structural algorithm by a factor of $\frac{24}{25}d^{-1}$. 
The reason for the factor of $\frac{24}{25}$ is the following variant of Lemma \ref{lem:alpha_bound}:
\begin{lemma}\label{lem:alpha_bound_sampling}
For $\beta \leq \frac{1}{50\mathcal{R}d^2}$ and for the structural algorithm with $N = 25$, $\alpha_j \leq \frac{24}{25}$, for any $j$.
\end{lemma}
\begin{proof}
Since $\beta \leq \frac{1}{50 \mathcal{R}d^2} \leq \frac{1}{50 \mathcal{R}d} = \frac{1}{2N\mathcal{R}d}$, Lemma \ref{lem:alpha_bound} implies that 
\[\alpha_j \leq \frac{24}{N}\left(1 - \frac{1}{\mathcal{R}}\right)^{|T_j| - b_j}\]
Note that since $b_j = |T_j - S_j|$, $|T_j| - b_j \geq 0$. It follows that $\alpha_j \leq \frac{24}{25}$.
\end{proof}

To use the structural algorithm as a subroutine for sampling, it will be necessary that $\alpha_j$ is bounded away from 1, which is why the choice of a slightly larger temperature bound ensures that $\alpha_j \leq 1-\frac{1}{25}$.

The additional factor of $d^{-1}$ allows for the efficient approximation of the partition function of the Hamiltonians $H^{S}$ for any $S \subseteq [n]$.
To do so, we cite Theorem 3.7 from \cite{bakshi2024high}. We remark that this theorem, although originally proved for the spin setting, can be applied without modification to the fermionic setting because the terms in the fermionic Hamiltonian, having even parity, are still ``bosonic'' in the sense that non-overlapping terms commute with each other.

\begin{lemma}\label{lem:eff_partition}
Say $H$ is an $n$-mode fermionic Hamiltonian with fixed locality $\mathcal{R}$ and degree $d$. For $\beta < \frac{1}{e(e+1)(1+e(d - 1))d}$, the partition function $\text{tr}(e^{-\beta H})$ may be estimated to multiplicative precision $\epsilon$ in time $\text{poly}(n, \epsilon^{-1})$. 
\end{lemma}

\begin{corollary}\label{cor:partition_estimate}
When $\beta \leq \frac{1}{50\mathcal{R}d^2}$, the partition function of an $n$-mode fermionic Hamiltonian $H$ with fixed locality $\mathcal{R}$ and degree $d$ as well as the quantity $\text{tr}(e^{-\beta H })$ can be estimated to multiplicative precision $\epsilon$ in time $\text{poly}(n, \epsilon^{-1})$. 
\end{corollary}
\begin{proof}
For $\mathcal{R} \geq 2$, 
\begin{equation}
\frac{1}{50\mathcal{R}d^2} \leq \frac{1}{100d^2} < \frac{1}{e(e+1)(1+e(d-1))d}
\end{equation}
for any $d \geq 1$. 
\end{proof}

For this choice of $\beta$, therefore, we may use the partition function and the structural algorithm together to implement a Gibbs sampling algorithm. 
In particular, we will make use of a subroutine of the algorithm, which is essentially just one of the $2n$ steps in Section \ref{subsec:struc_alg_description} with $N=25$:
\begin{algorithm}[Pinning one site]
\label{alg:pinning_subroutine}
\mbox{}\\

\textbf{Input:} Subset of terms $S \subseteq [2n]$ for unpinned sites and iteration $j$. Pinned unnormalized Gaussian states $\sigma$ and unpinned state $I + \alpha \Gamma.$

\textbf{Output:} $(j', S', \sigma', \Gamma', \alpha')$, representing the pinned version of the input states at the next iteration.
\begin{enumerate}
    % \item Choose $i$ to be the minimal index in the support of $\Gamma$; if $\Gamma$ has empty support, choose the minimal $i$ in $S$. Define $S' = S - \{i\}$. 
    \item If $\text{supp}(\Gamma) \cap S$ is nonempty, set $i$ as the minimal index in $\text{supp}(\Gamma) \cap S$. Otherwise, let $i$ be the minimal index in $S$.
    Define $S' = S - \{i\}$. 
    Define $$M = e^{\beta H^{S'/2}}e^{-\beta H^{S/2}}.$$
    \item Sample $I + \beta_1 \Lambda_1$ for $\beta_1 = (25\mathcal{R})^{-t_1}$ from the probability distribution in Lemma \ref{lem:fermion_convexcomb}, so that
    \[
    \mathbb{E}\left[I + \beta_1 \Lambda_1\right] = M =  e^{\beta H^{S'/2}}e^{-\beta H^{S/2}}.
    \]
    \item Sample $I + \beta_2 \Lambda_2$ for $\beta_2 = (25\mathcal{R})^{-t_2}$ independently, so that
    \[
    \mathbb{E}\left[I + \beta_2 \Lambda_2\right] = M^\dag.
    \]
    \item Set $\sigma' \gets \sigma$ and $(\Gamma', \alpha')$ as:
        \begin{enumerate}
            \item $\bigl(\Gamma, (1-\frac{1}{\mathcal{R}})^{-1}\alpha_{j}
            \bigl)$ with probability $1 - \frac{1}{\mathcal{R}}$
            \item $\left(\Lambda_1, (6\mathcal{R})\beta_1\right)$ with probability $\frac{1}{6\mathcal{R}}$
            \item $\left(\Lambda_2, (6\mathcal{R})\beta_2\right)$ with probability $\frac{1}{6\mathcal{R}}$
            \item $\left(\Lambda_1\Gamma, (6\mathcal{R})\beta_1\alpha_{j}\right)$ with probability $\frac{1}{6\mathcal{R}}$
            \item $\left(\Gamma\Lambda_2, (6\mathcal{R})\alpha_{j} \beta_2\right)$ with probability $\frac{1}{6\mathcal{R}}$
            \item $\left(\Lambda_1\Lambda_2, (6\mathcal{R})\beta_1 \beta_2\right)$ with probability $\frac{1}{6\mathcal{R}}$
            \item $\left(\Lambda_1\Gamma\Lambda_2, (6\mathcal{R})\beta_1\alpha_{j}  \beta_2\right)$ with probability $\frac{1}{6\mathcal{R}}$
        \end{enumerate}
    \item If $\Gamma'$ is not Hermitian:
    \begin{itemize}
        \item $\Gamma' \leftarrow 0, \alpha' \leftarrow 0$.
    \end{itemize}
    \item  If there are two indices $k, l \in S_j^c$ such that $k, l \in \text{supp}\left(\Gamma'\right)$, i.e. $\Gamma' =\Gamma^- \cdot (i\gamma_k \gamma_l)$: 
    \begin{enumerate}
        \item $\Gamma' \leftarrow \Gamma^-$, $\sigma' \leftarrow (I + i\gamma_k\gamma_l) \cdot \sigma'$ with probability $\frac{1}{2}$
        \item $\Gamma' \leftarrow -\Gamma^-$, $\sigma' \leftarrow (I - i\gamma_k\gamma_l) \cdot \sigma'$ with probability $\frac{1}{2}$
    \end{enumerate}
    \item If $\Gamma' \in \{\pm I\}$:
        \begin{itemize}
            \item $\sigma' \leftarrow (I + \alpha' \Gamma')\sigma'$, $\Gamma' \leftarrow 0$,  and $\alpha' \leftarrow 0$. 
        \end{itemize}
    \item $j \gets j + 1$.
\end{enumerate}
\end{algorithm}
Recall that the structural algorithm samples unnormalized Gaussian states $\sigma_i$ with probability $p_i$ such  that 
\[\sum_i p_i\sigma_i = e^{-\beta H}.\]
As discussed in the overview, a classical sampling algorithm must sample the operators $\rho_i=\frac{\sigma_i}{\text{tr}(\sigma_i)}$ according to the distribution 
\[q_i = \frac{p_i\text{tr}(\sigma_i)}{\text{tr}(e^{-\beta H})}.\]
Preparing the operators $\rho_i$ with respect to the distribution $p_i$ is easy, because the classical description of $\sigma_i$ can easily be renormalized to obtain a classical description of $\rho_i$. 
Thus, the main obstacle is converting this naive sampler to one that can sample $\rho_i$ with the correct probabilities, $q_i$, yielding a classical algorithm that can prepare 
\[\sum_{i} q_i \rho_i = \frac{e^{-\beta H}}{\text{tr}(e^{-\beta H})}.\]

\subsection{Sample Tree}\label{subsec:sample_tree}
In general, reducing a sampler for a distribution $q$ to a sampler for a distribution $p$ is a daunting task. 
However, there is a natural tree structure, $\mathcal{T}$, associated with these distributions. 
The tree has a root node with associated data $(0, S_0, \sigma_0, \Gamma_0,\alpha_0) = (0, [2n], I, 0, 0)$, and the children of a given node correspond to all possible random choices in the Algorithm \ref{alg:pinning_subroutine}, with the algorithm's output $(j+1, S_{j+1}, \sigma_{j+1}, \Gamma_{j+1}, \alpha_{j+1})$ as the associated data.
For a node $v$, we will denote its associated quantities by $j_v, S_v, \sigma_v, \Gamma_v$, and $\alpha_v$. 
For any internal node $v$ and any (potentially unnormalized) distribution on the leaves $r$, we use the notation $\mathcal{L}_\mathcal{T}(v)$ to denote the leaves below a given node $v$, and $r_v = \sum_{w\in \mathcal{L}_T(v)}r_w$.
% $r_v = \sum_{w \text{ leaf of }v}r_w$.
% Letting $\mathcal{L}_\mathcal{T}(v)$ be the leaves below a given node $v$, $r_v = \sum_{w\in \mathcal{L}_T(v)}r_w$. \YT{$r_v$ is doubly defined}
Our structural algorithm is a walk down this tree from $j = 0$ to $j = 2n$, by repeated application of Algorithm \ref{alg:pinning_subroutine}, and the probability distribution $p$ on the leaves is the probability of reaching a given leaf applying this algorithm.
$p_v$ for a general internal node is therefore the probability of reaching the node $v$. 
One final piece of notation is that $\mathcal{C}_\mathcal{T}(v)$ will be used to denote the set of children of a given node $v$. 

A reduction between sampling algorithms of two distributions on the leaves of a tree is possible when additional conditions are satisfied on the internal nodes of the tree. 
In particular, to convert a sampling algorithm for $p_i$ to a sampling algorithm for $q_i$, we make use of Theorem 5.8 in \cite{bakshi2024high}.
We replace a certain approximation ratio in their formulation, 10, with an arbitrary constant $\gamma$. 
We note that proof of the theorem in \cite{bakshi2024high} generalizes to this case as well.
\begin{theorem}\label{thm:simulate_distr}
Assume that $p$ and $q$ are two probability distributions on the leaves of a tree $\mathcal{T}$ of depth $d$, where each node has at most $k$ children. 
Let $\hat{p}$ and $\hat{q}$ be any unnormalized distributions that are equal to $p$ and $q$ up to some constant, i.e. $\hat{p}_v = Cp_v$ and $\hat{q}_v = Dq_v$ for any node $v$ and some constants $C, D \neq 0$. 
Assume also the following four conditions:
\begin{enumerate}
\item For any internal node $v$, the values $\frac{\hat{q}_v}{\hat{p}_v}$ can be calculated efficiently up to constant multiplicative error.
\item For any leaf $v$, the values $\frac{\hat{q}_v}{\hat{p}_v}$ can be efficiently calculated exactly.
\item For any internal node $v$, the leaves $v_k$ of $v$ can be sampled efficiently according to the distribution $\frac{p_{v_k}}{p_v}$.
\item Finally, for any node $v$ and its child $w$,
\[\left|\log\left(\frac{\hat{q}_v}{\hat{p}_v}\right)-\log\left(\frac{\hat{q}_w}{\hat{p}_w}\right)\right | \leq \log(\gamma)\]
for some constant $\gamma > 1$. 
\end{enumerate}

Then, there exists a $\text{poly}\left(d, \log(k), \log\left(\epsilon^{-1}\right)\right)$ algorithm to sample from $q$ up to total variation distance $\epsilon$. 
\end{theorem}

Theorem \ref{thm:simulate_distr} works by performing a carefully chosen random walk on the tree $\mathcal{T}$. 
This walk is designed so that its marginal distribution on the leaves of the tree converges to the target distribution $q$, using only local information about the ratios $\hat{q}_v/\hat{p}_v$.
At each node of the tree, the walk chooses to move up, down, or stay in place based on local comparisons of these unnormalized probabilities. 
The assumption that the log-ratio $\log(\hat{q}_v/\hat{p}_v)$ varies slowly across parent-child pairs is crucial: if the changes were large, they could prevent the walk from effectively concentrating around the parts of the tree where $q$ is most dense. 
We refer readers to Algorithm 5.9 in \cite{bakshi2024high} for the technical details of the walk as well as its mixing time analysis. 

We will first establish that almost all of the conditions of Theorem \ref{thm:simulate_distr} can be satisfied for the sampling tree $\mathcal{T}$ of Algorithm \ref{alg:pinning_subroutine}. 
The analysis requires an estimate of how the trace of a positive semidefinite matrix is affected by perturbations. 
\begin{lemma}\label{lem:matrix_perturbation}
Assume that $P$ is a positive semidefinite operator, and $A$ is some operator such that $\lVert A \rVert_\infty \leq 1 - \epsilon$ for $0<\epsilon \leq 1$. 
Then \[\left|\log(\text{tr}(P(1+A))) - \log(\text{tr}(P))\right| \leq \text{log}\left(\epsilon^{-1}\right)\]
\end{lemma}
\begin{proof}
We may use Hölder's inequality to say that 
\begin{equation}
\begin{aligned}
|\text{tr}(PA)|&\leq \lVert PA\rVert_1 \\&\leq \lVert P\rVert_1  \lVert A\rVert_\infty\\&= \text{tr}(P)\lVert A\rVert_\infty
\end{aligned}
\end{equation}
using the property that the $1$-norm is the sum of the absolute values of the eigenvalues of an operator. 
From this, we deduce that 
\begin{equation}
\text{tr}(P)(I-\lVert A\rVert_\infty) \leq \text{tr}(P(1+A))  \leq \text{tr}(P)(I+\lVert A\rVert_\infty),
\end{equation}
and in particular by the bound on $\lVert A \rVert_\infty$:
\begin{equation}
\epsilon \text{tr}(P)  \leq \text{tr}(P(1+A))  \leq (2-\epsilon)\text{tr}(P).
\end{equation}
It follows that 
\begin{equation}
\left|\log(\text{tr}(P(1+A))) - \log(\text{tr}(P))\right| \leq \max(\log(\epsilon^{-1}), \log(2-\epsilon)) = \log(\epsilon^{-1}),
\end{equation}
where the last equality holds since $(2-\epsilon)\epsilon \leq 1$ for $0 < \epsilon \leq 1$. 
\end{proof}

To prove condition (1), it will also be important to track the expectation of the output of Algorithm \ref{alg:structural_alg} conditioned on reaching a given internal node of $\mathcal{T}$.
\begin{lemma}\label{lem:gen_expectation}
Let $A_v$ be the event of reaching an internal node $v$ in the structural algorithm. Then $\mathbb{E}[\sigma_{2n} |A_v] = e^{-\beta H ^{S_v}/2}(I + \alpha_v \Gamma_v)\sigma_ve^{-\beta H ^{S_v}/2}$. 
\end{lemma}
\begin{proof}
Say that the node $v$ corresponds to the $j$th step of the algorithm.
As in Lemma \ref{lem:fermion_expectation}, we can use Equation \ref{eq:rewritteninductiveexpectation}, along with the law of total expectations and properties of $\mathcal{F}$-measurable random variables:
\begin{align*}
&\mathbb{E}\left[M_j^{-1}\dots M_{2n}^{-1}(I+\alpha_{2n}\Gamma_{2n})\sigma_{2n}(M_{2n}^{\dag})^{-1} \dots (M_{j}^{\dag})^{-1}| \mathcal{F}_j\right] \\
&= \mathbb{E}\left[\mathbb{E}\left[M_j^{-1}\dots M_{2n-1}^{-1}M_{2n}^{-1}(I+\alpha_{2n}\Gamma_{2n})\sigma_{2n}(M_{2n}^{\dag})^{-1} (M_{2n-1}^{\dagger})^{-1}\dots (M_{j}^{\dag})^{-1}| \mathcal{F}_{2n-1}\right]| \mathcal{F}_j\right]
\\
&= \mathbb{E}\left[M_j^{-1}\dots M_{2n-1}^{-1}\mathbb{E}\left[M_{2n}^{-1}(I+\alpha_{2n}\Gamma_{2n})\sigma_{2n}(M_{2n}^{\dag})^{-1} | \mathcal{F}_{2n-1}\right](M_{2n-1}^{\dagger})^{-1}\dots (M_{j}^{\dag})^{-1}| \mathcal{F}_j\right]\\
&= \mathbb{E}\left[M_j^{-1}\dots (I+\alpha_{2n-1}\Gamma_{2n-1})\sigma_{2n-1}\dots (M_{j}^{\dag})^{-1}| \mathcal{F}_j\right]
\\
&= \dots = \mathbb{E}\left[(I + \alpha_j \Gamma_j)\sigma_j | \mathcal{F}_j\right]\\\
&= (I + \alpha_j \Gamma_j)\sigma_j. 
\end{align*}
Note that $M_{2n} \dots M_j = e^{-\beta H^{S_j}/2}$ and $\alpha_{2n} = 0$, the left hand side of the above expression is 
\begin{equation}
\mathbb{E}\left[e^{\beta H^{S_j}/2}(I + \alpha_{2n}\Gamma_{2n})\sigma_{2n}e^{\beta H^{S_j}/2} | \mathcal{F}_j\right] = e^{\beta H^{S_j}/2}\mathbb{E}\left[\sigma_{2n}| \mathcal{F}_j\right]e^{\beta H^{S_j}/2},
\end{equation}
since $S_j$ depends only on randomness up until step $j$. 
We conclude that
\begin{equation}
\mathbb{E}\left[\sigma_{2n}| \mathcal{F}_j\right] = e^{-\beta H ^{S_j}/2}(I + \alpha_j \Gamma_j)\sigma_je^{-\beta H ^{S_j}/2}
\end{equation}
Since $A_v$ is an event based only on random choices up to the $j$th step, it is true that
\[\mathbb{E}[\cdot | A_v] = \mathbb{E}[\mathbb{E}[\cdot | \mathcal{F}_j] | A_v],\]
from which it follows that
\[\mathbb{E}[\sigma_{2n} | A_v] = e^{-\beta H ^{S_v}/2}(I + \alpha_v \Gamma_v)\sigma_ve^{-\beta H ^{S_v}/2}.\]
\end{proof}

\begin{lemma} \label{lem:cond_1_2_satisfied}
Let $\mathcal{T}$ be the tree associated with Algorithm \ref{alg:pinning_subroutine} with the associated distributions $p$ and $q$, and let $\hat{p}_v = p_v$, $\hat{q}_v = \text{tr}(e^{-\beta H})q_v = \text{tr}(\sigma_v)p_v$ be the corresponding unnormalized distributions for any leaf $v$. 
Then conditions (1) and (2) of Theorem \ref{thm:simulate_distr} are satisfied. 
\end{lemma}
\begin{proof}
See that
\begin{align}
\frac{\hat{q}_v}{\hat{p}_v} &= \frac{\sum_{w\in \mathcal{L}_T(v)}\hat{q}_w}{\sum_{w\in \mathcal{L}_T(v)}p_w}\\
&= \frac{\sum_{w\in \mathcal{L}_T(v)}\text{tr}(\sigma_w)p_w}{\sum_{w\in \mathcal{L}_T(v)} p_w }\\
&= \mathbb{E}[\text{tr}(\sigma_w) \,|\, \text{structural algorithm \ref{alg:structural_alg} reaches }v]\\
&= \text{tr}(\mathbb{E}[\sigma_w \,|\, \text{structural algorithm \ref{alg:structural_alg} reaches }v])\\
&=\text{tr}(e^{-\beta H^{S_v}/2}(I+\alpha_{v}\Gamma_v)\sigma_v e^{-\beta H^{S_v}/2}) \label{eq:cond_exp}\\
&= \text{tr}(\sigma_ve^{-\beta H^{S_v}}(I+\alpha_v \Gamma_v))
\end{align}
where Eq. ~(\ref{eq:cond_exp}) uses Lemma \ref{lem:gen_expectation}. 
The state $\sigma$ acts on different modes than $e^{-\beta H^{S_v}}$, as we have established its suppport is contained in $S_v^c$. 
As a consequence,
\begin{equation}
\text{tr}\left(\sigma_v e^{-\beta H^{S_v}}\right) = \text{tr}(\sigma_v)\text{tr}\left(e^{-\beta H^{S_v}}\right).
\end{equation}
% \YT{This seems to be assuming that $\sigma_v$ is normalized, which does not seem to be true from Step 7 on Page 24? Though $\text{tr}(\sigma_v)$ can still be computed so we can let $P=\sigma_v e^{-\beta H^{S_v}}$ below.}
For the same reason, $\sigma_v e^{-\beta H^{S_v}}$ is positive semidefinite. 
Setting $P = \sigma_ve^{-\beta H^{S_v}}$ and $A = \alpha_v \Gamma_v$, note that $\lVert A \rVert_\infty = \alpha_v$. 
Applying Lemma \ref{lem:matrix_perturbation}, 
\begin{align*}
\left\lvert\log\left(\frac{\hat{q}_v}{\hat{p}_v}\right) - \log(\text{tr}(P))\right\rvert &= \left\lvert\log\left(\text{tr}(P(I+A))\right) - \log(\text{tr}(P))\right\rvert\\
&\leq \log((1-\alpha_v)^{-1}) \\
&\leq \log(25),
\end{align*}
where the last inequality holds by Lemma \ref{lem:alpha_bound_sampling}. 
We note that $\text{tr}(P) = \text{tr}(\sigma_v)\text{tr}\left(e^{-\beta H^{S_v}}\right)$, and $\text{tr}(\sigma_v)$ is known exactly. Since $H^{S_v}$ has locality and degree bounded by the constant locality $\mathcal{R}$ and degree $d$ of $H$, Corollary \ref{cor:partition_estimate} may therefore be applied to calculate its $\frac{\hat{q}_v}{\hat{p}_v}$ up to multiplicative error $25$. 
This establishes condition (1).

For condition (2), note that associated to any leaf $v$ is the state $\sigma_v$, whose classical representation includes the trace $\text{tr}(\sigma_v)$. 
It follows that $\frac{\hat{q}_v}{\hat{p}_v} = \text{tr}(\sigma_v)$ can be calculated exactly for any leaf $v$. 
Thus, condition (2) can be satisfied as well. 
\end{proof}

Condition (3) does not hold true for the tree $\mathcal{T}$. 
One simple way to see this is that in every application of Algorithm \ref{alg:pinning_subroutine}, there are infinite possible outcomes of the random choices, due to the sampling from probability distributions with infinite support in steps 2 and 3.  
Furthermore, 
Condition (4) is satisfiable, using the following lemma. 
After establishing this condition, we will design a subtree $\mathcal{T}'$ of $\mathcal{T}$ for which all four conditions may be attained, and ensure that the parameter $k$ in Theorem \ref{thm:simulate_distr} is such that $\log(k) \in \text{poly}(n)$. 
\begin{lemma}\label{lem:gen_decomp}
For any $S_i$ and $S_{i-1}$ as in Lemma \ref{lem:fermion_taylordecomp}, if $\beta \leq \frac{1}{2Cd}$ for some value $C > 1$, then $e^{-\beta H^{S_{i}}}e^{\beta H^{S_{i-1}}}= I + A$ where $\lVert A \rVert_\infty \leq \frac{1}{C-1}$.
\end{lemma}
\begin{proof}
Lemma \ref{lem:fermion_taylordecomp} decomposes $M_i$ as $I + \sum_{t>0}\sum_j c_{j, t} \Gamma_{j,t}$ such that $\sum_j c_{j, t} \leq C^{-t} $. 
Defining $A = \sum_{t>0} \sum_j c_{j, t} \Gamma_{j,t}$, we find that
\begin{equation}
\lVert A\rVert_\infty \leq \sum_{t > 0}\sum_j c_{j, t} \leq \sum_{t>0}C^{-t} = \frac{1}{C-1}. 
\end{equation}
\end{proof}
\begin{lemma} \label{lem:cond_4_satisfied}
Let $\mathcal{T}$ be the tree associated with Algorithm \ref{alg:pinning_subroutine} with the associated distributions $p$ and $q$, and let $\hat{p}_v = p_v$, $\hat{q}_v = \text{tr}(e^{-\beta H})q_v = \text{tr}(\sigma_v)p_v$ be the corresponding unnormalized distributions for any leaf $v$. 
Then condition (4) of Theorem \ref{thm:simulate_distr} is satisfied, with constant $\gamma_\mathcal{T} = 16000$. 
\end{lemma}

Next, we verify condition (4) of Theorem \ref{thm:simulate_distr}, that 
\[\left|\log\left(\frac{\hat{q}_v}{\hat{p}_v}\right)-\log\left(\frac{\hat{q}_w}{\hat{p}_w}\right)\right | \leq \log(\gammat),\]
for the constant $\gammat = 16000$.
Given that $\frac{\hat{q}_v}{\hat{p}_v} = \text{tr}(e^{-\beta H^{S_v}} (I + \alpha_v \Gamma_v) \sigma_v)$, from the proof of Lemma~\ref{lem:cond_1_2_satisfied},
we need only show that
\begin{equation}
    \left|\log\left(\text{tr}(e^{-\beta H^{S_v}}(I+\alpha_v \Gamma_v)\sigma_v)\right) - \log\left(\text{tr}(e^{-\beta H^{S_w}}(I+\alpha_w \Gamma_w)\sigma_w)\right)\right| = O(1).
\end{equation}
By an application of the triangle inequality, the following three relations suffices:
\[\left|\log\left(\text{tr}(\sigma_ve^{-\beta H^{S_v}}(I+\alpha_v \Gamma_v))\right) - \log\left(\text{tr}(\sigma_ve^{-\beta H^{S_v}})\right)\right| = O(1),\]
\[\left|\log\left(\text{tr}(\sigma_we^{-\beta H^{S_w}}(I+\alpha_w \Gamma_w))\right) - \log\left(\text{tr}(\sigma_we^{-\beta H^{S_w}})\right)\right| = O(1),\]
\[\left|\log\left(\text{tr}(\sigma_v)\text{tr}( e^{-\beta H^{S_v}})\right) - \log\left(\text{tr}(\sigma_w)\text{tr}(e^{-\beta H^{S_w}})\right)\right| = O(1).\]
The first two have already been shown in Lemma \ref{lem:cond_1_2_satisfied} in the proof of condition (1), each with the bound $25 = O(1)$. 

For the third relation, applying Lemma \ref{lem:gen_decomp} to $S_v$ and $S_w$ and using the triangle inequality, we may rewrite the inequality as
\[\left|\log\left(\text{tr}(e^{-\beta H^{S_w}}(I+A))\right) - \log\left(\text{tr}(e^{-\beta H^{S_w}})\right)\right| + \left|\log\left(\text{tr}(\sigma_v)\right) - \log\left(\text{tr}(\sigma_w)\right)\right| =O(1)\]
for some $A$. 
The second summand is at most $\log(25)$, since in any step $\sigma_v$ is multiplied by a commuting operator with trace 1 or trace $1 \pm \alpha_v$. 
This latter value lies in $[\frac{1}{25}, 1 + \frac{24}{25}]$, and therefore the logarithm is bounded in absolute value by $\log(25)$. 
Since $\beta \leq \frac{1}{50\mathcal{R}d^2} \leq \frac{1}{100d}$, we can take $C=25$, so $\lVert A \rVert_\infty \leq \frac{1}{49}$. 
Lemma \ref{lem:matrix_perturbation} applies again since $e^{-\beta H S_w}$ is positive, yielding an upper bound of $\log(25 \cdot 49 \cdot 48^{-1})$, including the second summand. 

In particular, we have shown that 
\begin{equation}
    \begin{aligned}
        &\left|\log\left(\text{tr}(e^{-\beta H^{S_v}}(I+\alpha_v \Gamma_v))\right) - \log\left(\text{tr}(e^{-\beta H^{S_w}}(I+\alpha_w \Gamma_w))\right)\right| \\&\leq \log\left(\left(\frac{1}{25}\right)^{-3}\cdot \left(\frac{48}{49}\right)^{-1}\right) \\
        &\leq \log(\gammat),
    \end{aligned}
\end{equation}
as desired.

\subsection{Truncated Sample Tree}\label{subsec:truncated_sample_tree}
\begin{theorem}[efficient algorithm]\label{thm:efficient_gibbs_sampling}
Assume that $H$ has locality $\mathcal{R}$ and has degree $d$, and $\beta \leq \frac{1}{50\mathcal{R}d^2}$. 
Then there exists a classical algorithm that can sample the Gibbs state to trace distance $\epsilon$ in time $\text{poly}(n, \log(\epsilon^{-1}))$.
\end{theorem}
To prove this theorem, we will create an approximate version of Algorithm \ref{alg:pinning_subroutine}, such that the corresponding tree $\mathcal{T}'$ will satisfy the conditions of Theorem \ref{thm:simulate_distr}. 
The following lemmas yields an approximate version of a sampling step in the algorithm. 

\begin{lemma}\label{lem:eff_sampl}
    There exists a $\text{poly}\left(\delta^{-1}\right)$-time algorithm that samples from a distribution $\mu'$ of outcomes $(t_k, I + (25\mathcal{R})^{-t_k}\Gamma_k)$ of Lemma \ref{lem:fermion_convexcomb}, which is precisely the distribution $\mu$ from Lemma \ref{lem:fermion_convexcomb} conditioned on an event $A$ with $\mathbb{P}_\mu[A] \geq 1 - \delta$, i.e. 
    \begin{equation}
    \mu'_X = \begin{cases}
    \displaystyle \frac{\mu_X}{\mathbb{P}_\mu[A]} & \text{if } X \in A \\
    0 & \text{if } X \notin A
    \end{cases}.
    \end{equation}
    Specifically, $A$ is the event that $t_k \leq t_\text{max}$ for threshold $t_\text{max} = \lceil \log_2\left(\delta^{-1}\right)\rceil$. 
    Moreover, the distribution $\mu'$ has support with size $\text{exp}(\text{polylog}\left(\delta^{-1}\right))$.
\end{lemma}
\begin{proof}
% \YT{what even $A$ is should be stated in the theorem} 
%Let the event $A$ be the event that $t_k \leq t_\text{max}$ for threshold $t_\text{max} = \lceil \log_2\left(\delta^{-1}\right)\rceil$. 
The distribution $\mu$ corresponds to the probabilities in the convex combination
\begin{equation}
M_i =   \sum_{t>0}\left(\left(2^{-t} - c_t(N\mathcal{R})^t\right)I + \sum_j c_{j,t}(N\mathcal{R})^t(I+ (N\mathcal{R})^{-t} \Gamma_{j,t})\right),
\end{equation}
where $\Gamma_{j,t}$ are $\mathcal{R}t$-local Majorana operators and $c_{j,t} \in \mathbb{R}^+.$
We will show that the probabilities in the convex combination
\begin{equation}
\sum_{t=1}^{t_{\text{max}}}\left(\frac{\left(2^{-t} - c_t(N\mathcal{R})^t\right)}{1-2^{-t_\text{max}}}I + \sum_j \frac{c_{j,t}(N\mathcal{R})^t}{1-2^{-t_\text{max}}}(I+ (N\mathcal{R})^{-t} \Gamma_{j,t})\right)
\end{equation}
can be sampled efficiently, and this is precisely the original distribution conditioned on the event $t_k \leq t_\text{max}$, which occurs with probability $1-2^{-t_\text{max}} \geq 1 - \delta$. 

This algorithm to perform this sampling will proceed in two steps. 
Firstly, it samples a value of $t = 1 , \dots, t_\text{max}$ from the distribution $\frac{2^{-t}}{1 - 2^{-t_\text{max}}}$, which may be done exactly and efficiently. 
Then, given the sampled $t$, it must sample $I$ with probability $1 - c_t(2N\mathcal{R})^t$ and each $I + (N\mathcal{R})^{-t}\Gamma_{j,t}$ with probability $c_{j, t}(2N\mathcal{R})^t$.
However, Corollary \ref{cor:sampl_cjt} is able to sample each $\Gamma_{j,t}$ with probabilities $p_{j,t} \geq c_{j,t}(2N\mathcal{R})^t$, in time. 
After sampling an operator $\Gamma_{j,t}$, the algorithm can further sample from $\{0, \Gamma_{j,t}\}$ with probabilities $1 - \frac{p_{j,t}}{(2N\mathcal{R})^{-t}}$ and $\frac{p_{j,t}}{c_{j, t}(2N\mathcal{R})^{-t}}$. We are now sampling each $\Gamma_{j,t}$ with probability $c_{j, t}(2N\mathcal{R})^{-t}$ and $0$ with the remaining probability.
Returning $I + (2N\mathcal{R})^{-t}\Gamma_{j,t}$ or $I$ correspondingly gives the desired algorithm.

The above algorithm chooses a value for $t \leq t_\text{max}$, and then applies Corollary \ref{cor:sampl_cjt} for $t$, which is efficient in $\log\delta^{-1}$ and $n$. 
Furthermore, Corollary \ref{cor:sampl_cjt} guarantees that the number of possible terms $(c_{j,t}, \Gamma_{j,t})$ for a given $t$ is at most $(t+1)! \cdot (d + 1)^t.$
Since we truncate the expansion at $t_{\text{max}} = \lceil \log_2 (\delta^{-1})\rceil$, the total number of possible sampling outcomes is at most 
\begin{equation}
\sum_{t = 0}^{t_{\text{max}}} (t+1)! \cdot (d+1)^t \leq (t_{\text{max}} +1)\cdot (t_{\text{max}} + 1)! \cdot (d+1)^{t_{\text{max}}}.
\end{equation}
Because $\log(\delta^{-1})!$ is quasipolynomial in $\delta^{-1}$, $(t_\text{max}+1)! = \text{exp}(\text{polylog}(\delta^{-1}))$.
Likewise, $(d+1)^{t_\text{max}} = O(\delta^{-\log(d+1)}) = \text{poly}(\delta^{-1})$. 
We conclude that the support of the sampling distribution is $\text{exp}(\text{polylog}(\delta^{-1}))$.

\end{proof}

Using the above lemma, we can adapt Algorithm \ref{alg:pinning_subroutine} to make it efficient to implement, stated below:

\begin{algorithm}[Pinning one site, efficient]
\label{alg:pinning_subroutine_rounded}
\mbox{}\\

\textbf{Input:} Subset of terms $S \subseteq [2n]$ for unpinned sites and iteration $j$. Pinned unnormalized Gaussian states $\sigma$ and unpinned state $I + \alpha \Gamma.$ Desired trace distance $\epsilon$ between the sampled state and the Gibbs state, for $\epsilon \geq 0.$

\textbf{Output:} $(j', S', \sigma', \Gamma', \alpha')$, representing the pinned version of the input states at the next iteration.
\begin{enumerate}
    % \item Choose $i$ to be the minimal index in the support of $\Gamma$; if $\Gamma$ has empty support, choose the minimal $i$ in $S$. Define $S' = S - \{i\}$. 
    \item If $\text{supp}(\Gamma) \cap S$ is nonempty, set $i$ as the minimal index in $\text{supp}(\Gamma) \cap S$. Otherwise, let $i$ be the minimal index in $S$.
    Define $S' = S - \{i\}$. 
    Define $$M = e^{\beta H^{S'/2}}e^{-\beta H^{S/2}}.$$
    \item Sample $I + \beta_1 \Lambda_1$ for $\beta_1 = (25\mathcal{R})^{-t_1}$ from the probability distribution in Lemma \ref{lem:eff_sampl} with $\delta = \frac{\epsilon}{16n(\gammat+2)}$ for $\gammat$ as defined in Lemma~\ref{lem:cond_4_satisfied}, so that
    \[
    \mathbb{E}\left[I + \beta_1 \Lambda_1\right] = M =  e^{\beta H^{S'/2}}e^{-\beta H^{S/2}}.
    \]
    \item Sample $I + \beta_2 \Lambda_2$ for $\beta_2 = (25\mathcal{R})^{-t_2}$ independently from the same distribution, so that
    \[
    \mathbb{E}\left[I + \beta_2 \Lambda_2\right] = M^\dag.
    \]
    \item Set $\sigma' \gets \sigma$ and $(\Gamma', \alpha')$ as:
        \begin{enumerate}
            \item $\bigl(\Gamma, (1-\frac{1}{\mathcal{R}})^{-1}\alpha_{j}
            \bigl)$ with probability $1 - \frac{1}{\mathcal{R}}$
            \item $\left(\Lambda_1, (6\mathcal{R})\beta_1\right)$ with probability $\frac{1}{6\mathcal{R}}$
            \item $\left(\Lambda_2, (6\mathcal{R})\beta_2\right)$ with probability $\frac{1}{6\mathcal{R}}$
            \item $\left(\Lambda_1\Gamma, (6\mathcal{R})\beta_1\alpha_{j}\right)$ with probability $\frac{1}{6\mathcal{R}}$
            \item $\left(\Gamma\Lambda_2, (6\mathcal{R})\alpha_{j} \beta_2\right)$ with probability $\frac{1}{6\mathcal{R}}$
            \item $\left(\Lambda_1\Lambda_2, (6\mathcal{R})\beta_1 \beta_2\right)$ with probability $\frac{1}{6\mathcal{R}}$
            \item $\left(\Lambda_1\Gamma\Lambda_2, (6\mathcal{R})\beta_1\alpha_{j}  \beta_2\right)$ with probability $\frac{1}{6\mathcal{R}}$
        \end{enumerate}
    \item If $\Gamma'$ is not Hermitian:
    \begin{itemize}
        \item $\Gamma' \leftarrow 0, \alpha' \leftarrow 0$.
    \end{itemize}
    \item  If there are two indices $k, l \in S_j^c$ such that $k, l \in \text{supp}\left(\Gamma'\right)$, i.e. $\Gamma' =\Gamma^- \cdot (i\gamma_k \gamma_l)$: 
    \begin{enumerate}
        \item $\Gamma' \leftarrow \Gamma^-$, $\sigma' \leftarrow (I + i\gamma_k\gamma_l) \cdot \sigma'$ with probability $\frac{1}{2}$
        \item $\Gamma' \leftarrow -\Gamma^-$, $\sigma' \leftarrow (I - i\gamma_k\gamma_l) \cdot \sigma'$ with probability $\frac{1}{2}$
    \end{enumerate}
    \item If $\Gamma' \in \{\pm I\}$:
        \begin{itemize}
            \item $\sigma' \leftarrow (I + \alpha' \Gamma')\sigma'$, $\Gamma' \leftarrow 0$,  and $\alpha' \leftarrow 0$. 
        \end{itemize}
    \item $j \gets j + 1$.
\end{enumerate}
\end{algorithm}
Indeed, the rest of the algorithm consists of products of scalars and Majorana strings, which can be encoded efficiently. 
For this new algorithm, we may now define a new tree $\mathcal{T}'$, where again the root node is $(0, S_0, \sigma_0, \Gamma_0, \alpha_0)$, and the children of a given node are obtained by an application of Algorithm \ref{alg:pinning_subroutine_rounded}. 
Note that $\mathcal{T}'$ can be naturally identified as a subtree of $\mathcal{T}$, because the two algorithms are identical except that the sampling distribution in Lemma \ref{lem:eff_sampl} has support contained in that of Lemma \ref{lem:fermion_convexcomb}. 
With this algorithm, we may define a distribution $p'$ on the leaves, where $p'_w$ is the probability of reaching the node $w$ from the root with a repeated application of Algorithm \ref{alg:pinning_subroutine_rounded}. 
Let $\hat{p}'_w = p'_w$ for any leaf. 
We may also define an unnormalized distribution $\hat{q}'$ on the leaves by 
\begin{equation}
\hat{q}'_w = \text{tr}(\sigma_w)\hat{p}'_w.
\end{equation}
Here, $\hat{q}'$ is being defined in direct analogy to $\hat{q}$.
Then, $q'$ is simply its normalization, i.e. 
\begin{equation}
q'_v = \frac{\hat{q}'_v}{\sum_{w\in \mathcal{L}_{\mathcal{T}'}(r)} \hat{q}'_v}= \frac{\hat{q}'_v}{\hat{q}'_r},
\end{equation}
where $r$ denotes the root node. 

Using Lemma \ref{lem:eff_sampl}, we can establish a relationship between $p_v$ and $p_v'$. 
In particular, for any vertex $v$ in the tree $\mathcal{T}'$, their difference is at most $Cp_v'$ for an appropriate constant that is determined by our choice of $\delta$. 
This bound is obtained by using a similar bound on the conditional distributions of the children of $v$, specifically $\frac{p_w'}{p_v'}$ and $\frac{p_w}{p_v}$. 
\begin{lemma}\label{lem:p_p'_relationship}
For the distribution $p'$ obtained from Algorithm ~\ref{alg:pinning_subroutine_rounded}, and for any node $v$ and child $w$ in the tree $\mathcal{T}'$, $\left|\frac{p_w'}{p_v'} - \frac{p_w}{p_v}\right| \leq \frac{\epsilon}{8n(\gammat+2)}\left(\frac{p_w'}{p_v'}\right)$. 
Moreover, for any leaf node $v \in \mathcal{L}_{\mathcal{T}'}(r)$, $|p_v - p_v'| \leq \frac{\epsilon}{4(\gammat+2)}p_v'$. 
% \YT{need to introduce $\epsilon$ and $\gamma_\tau$. Maybe say $p'$ is the distribution obtained from Algorithm~\ref{alg:pinning_subroutine_rounded}.}
\end{lemma}
\begin{proof}
The only difference between Algorithm \ref{alg:pinning_subroutine} and Algorithm \ref{alg:pinning_subroutine_rounded} is that the sampling in steps 2 and 3 are from different distributions. By Lemma \ref{lem:eff_sampl}, the distribution in Algorithm \ref{alg:pinning_subroutine_rounded} is simply the distribution in Algorithm \ref{alg:pinning_subroutine} conditioned on an event of probability at least $1 - \delta$. 

There are two such events $A_1, A_2$ that occur in each application of Algorithm \ref{alg:pinning_subroutine_rounded}, more specifically in steps 2 and 3. 
As such, the conditional distribution $\frac{p_w'}{p_v'}$ for $w$, i.e. the children of a fixed node $v$, is precisely $\frac{p_w}{p_v}$ conditioned on the joint event $A_1 \cap A_2$, which happens as an event of probability at least $1-2\delta$ by a union bound.
We choose $\delta = \frac{\epsilon}{16n(\gammat+2)}.$
It follows that
\[\left|\frac{p_w}{p_v} - \frac{p_w'}{p_v'}\right| = \frac{p_w'}{p_v'} - \frac{p_w}{p_v} \leq \frac{p_w'}{p_v'}-\left(1-\frac{\epsilon}{8n(\gammat+2)}\right)\frac{p_w'}{p_v'} = \frac{\epsilon}{8n(\gammat+2)}\left(\frac{p_w'}{p_v'}\right), \]
where the first inequality holds since $\frac{p_w'}{p_v'} = \frac{p_w/p_v}{\mathbb{P}_{p}[A_1 \cap A_2]} \leq \frac{p_w/p_v}{1-\epsilon/8n(\gammat+2)}$. 

To traverse the entire tree to a leaf, there are $2n$ steps, which means we iteratively apply Algorithm \ref{alg:pinning_subroutine_rounded} $2n$ times.
Thus, the distribution $p_v'$ over the leaves $v$ is $p_v$ conditioned on an event of probability at least $1 - 2n \cdot 2\delta = 1 - \frac{\epsilon}{4(\gammat+2)}$ by a union bound as well.
For any leaf node $v \in \mathcal{L}_{\mathcal{T}'}(r)$, it follows by the same reasoning as the above that  
\[|p_v - p_v'| \leq  \frac{\epsilon}{4(\gammat+2)}p_v'. \]
\end{proof}

For our new tree $\mathcal{T}'$, there is also a natural bound on the number of children of any node. 
\begin{lemma}\label{lem:tree_bound}
The number of children of any node in $\mathcal{T}'$ is at most $k$ for $\log(k) = \text{poly}(n, \log(\epsilon^{-1}))$. 
\end{lemma}
\begin{proof}
The sampling algorithm in Algorithm \ref{alg:pinning_subroutine} involves three random choices. 
The first two are calls to the efficient sampling algorithm from Lemma \ref{lem:eff_sampl}, in step 2 and 3 of the algorithm respectively, which has support with size $\text{exp}(\text{poly}(\log (\delta^{-1})))$.
We will use a more crude bound of $\text{exp}(\text{poly}(\delta^{-1}))$ for the support to simplify the choice of $\delta$ below.
We choose $\delta = \frac{\epsilon}{16n(\gammat+2)}$, from which we get the number of possible outcomes in the form of $(t_i, I + \beta_{t_i} \Lambda_{t_i})$ for each of these two calls is bounded by $\text{exp}(\text{poly}(n, \log (\epsilon^{-1}))).$
% The algorithm's random choice involves a choice of what degree to sample from, from zero to $\lceil \delta^{-1} \rceil$.
% Up to this degree, at each step some local term of $H$ 
% Each of these choices involves 
The next two calls are random choices from seven choices of $(\Gamma, \sigma, \alpha)$ in step 4 and two options of $(\Gamma, \sigma, \alpha)$ in step 6, giving an additional factor of 14, which is constant. 
This gives us the final number of children of at most $k = \text{exp}(\text{poly}(n, \log (\epsilon^{-1}))).$
\end{proof}

We will now establish two important results that relate distributions on $\mathcal{T}$ to those on $\mathcal{T}'$, presented in Lemma \ref{lem:approximate_expectation} and Lemma \ref{lem:q_approx} below, which will each be important in the proof of Theorem~\ref{thm:efficient_gibbs_sampling}.
Lemma \ref{lem:approximate_expectation} is the most technical, bounding the difference between $\frac{\hat{q}_v}{\hat{p}_v}$ and $\frac{\hat{q}'_v}{\hat{p}'_v}$. 
We have previously seen $\frac{\hat{q}_v}{\hat{p}_v}$ to be the expected value of $\text{tr}(\sigma_w)$ for leaves $w$ below $v$ under the distribution $p_w$, so intuitively the lemma establishes how well the primed distribution approximates these conditional expectations. 
It has two roles: it will allow us to relate $\hat{q}_r$ to $\hat{q}'_r$, and bound the multiplicative distance between $\frac{\hat{q}_v'}{\hat{p}_v'}$ and $\frac{\hat{q}_w'}{\hat{p}_w'}$ for a node $v$ and its child $w$ (to satisfy condition (1)) .

\begin{lemma}\label{lem:approximate_expectation}
For any node $v \in \mathcal{T}'$,
\[\left|\frac{\hat{q}'_v}{\hat{p}'_v}- \frac{\hat{q}_v}{\hat{p}_v}\right| \leq \frac{ (2n-j_v)(\gammat+1)\epsilon}{8n(\gammat + 2) }\max\left(\frac{\hat{q}_v'}{\hat{p}_v'} , \frac{\hat{q}_v}{\hat{p}_v}\right). \]
In particular, 
\[\left|\log\left(\frac{\hat{q}'_v}{\hat{p}'_v}\right)- \log\left(\frac{\hat{q}_v}{\hat{p}_v}\right)\right| \leq \log\left(1+\frac{ (2n-j_v)(\gammat+1)\epsilon}{4n (\gammat+2)}\right).\]
\end{lemma}
\begin{proof}
We will prove this reduction by descending strong induction on the depth of $v$, with the base case of $j_v = 2n$.
In this case, by definition $\frac{\hat{q}_v'}{\hat{p}_v'} = \frac{\hat{q}_v}{\hat{p}_v} = \text{tr}(\sigma_v)$, so the result is immediate.
Now see that we may write $\frac{\hat{q}'_v}{\hat{p}'_v}$ as 
\[\frac{\hat{q}'_v}{\hat{p}'_v} = \sum_{w\in \mathcal{C}_{\mathcal{T}'}(v)} \left(\frac{\hat{q}'_w}{\hat{p}'_w}\right) \cdot \left( \frac{\hat{p}'_w}{\hat{p}'_v}\right) = \sum_{w\in \mathcal{C}_{\mathcal{T}'}(v)} \left(\frac{\hat{q}'_w}{\hat{p}'_w}\right) \cdot \left( \frac{\hat{p}_w}{\hat{p}_v}\right)+\sum_{w\in \mathcal{C}_{\mathcal{T}'}(v)} \left(\frac{\hat{q}'_w}{\hat{p}'_w}\right)  \left( \frac{\hat{p}'_w}{\hat{p}'_v}-\frac{\hat{p}_w}{\hat{p}_v}\right),\]
Likewise,  
\[\frac{\hat{q}_v}{\hat{p}_v} = \sum_{w\in \mathcal{C}_{\mathcal{T}}(v)} \left(\frac{\hat{q}_w}{\hat{p}_w}\right) \cdot \left( \frac{\hat{p}_w}{\hat{p}_v}\right).\]
The proof of Lemma \ref{lem:p_p'_relationship} states that the distribution $\frac{p'_w}{p'_v}$ is the distribution $\frac{p_w}{p_v}$ conditioned on an event with probability at least $1 - \frac{\epsilon}{8n(\gammat+2)}$.

By induction, we assume that for all children nodes $w \in \mathcal{C}_{\mathcal{T}'}(v)$ (which necessarily satisfy
$j_w = j_v + 1$), the following holds: 
\[\left|\frac{\hat{q}_w'}{\hat{p}_w'} - \frac{\hat{q}_w}{\hat{p}_w} \right| \leq \left(\frac{ (2n-j_w )(\gammat+1)\epsilon}{8n(\gammat+2)}\right)\max\left(\frac{\hat{q}_v'}{\hat{p}_v'} , \frac{\hat{q}_v}{\hat{p}_v}\right).\]
The base case for leaf nodes is clear because they are precisely equal, both $\text{tr}(\sigma_w)$. 
Now see that
\begin{align*}
&\left|\frac{\hat{q}_v'}{\hat{p}_v'} - \frac{\hat{q}_v}{\hat{p}_v} \right| \\&= 
\left| \sum_{w \in \mathcal{C}_{\mathcal{T}'}(v)} \left(\frac{\hat{q}_w'}{\hat{p}_w'} - \frac{\hat{q}_w}{\hat{p}_w}\right)\cdot \left(\frac{\hat{p}_w}{\hat{p}_v}\right) + \sum_{w\in \mathcal{C}_{\mathcal{T}'}(v)} \left(\frac{\hat{q}'_w}{\hat{p}'_w}\right)  \left( \frac{\hat{p}'_w}{\hat{p}'_v}-\frac{\hat{p}_w}{\hat{p}_v}\right) - \sum_{w \in \mathcal{C}_{\mathcal{T}}(v) \backslash \mathcal{C}_{\mathcal{T}'}(v)} \left(\frac{\hat{q}_w}{\hat{p}_w}\right)\cdot \left(\frac{\hat{p}_w}{\hat{p}_v}\right)\right|\\
&\leq \sum_{w \in \mathcal{C}_\mathcal{T}'(v)} \left|\frac{\hat{q}_w'}{\hat{p}_w'} - \frac{\hat{q}_w}{\hat{p}_w}\right|\cdot \left(\frac{\hat{p}_w}{\hat{p}_v}\right) + \frac{\epsilon}{8n(\gammat+2)} \left(\frac{\hat{q}_v'}{\hat{p}_v'}\right) + \frac{\gammat \epsilon}{8n(\gammat+2)}\left(\frac{\hat{q}_v}{\hat{p}_v    }\right)\\
&\leq \frac{(2n - j_v - 1)(\gammat+1)\epsilon}{8n(\gammat+2)}\max\left(\frac{\hat{q}_v'}{\hat{p}_v'} , \frac{\hat{q}_v}{\hat{p}_v}\right) + \frac{\epsilon}{8n(\gammat+2)}\left(\frac{\hat{q}_v'}{\hat{p}_v'}\right) + \frac{\gammat \epsilon}{8n(\gammat+2)}\left(\frac{\hat{q}_v}{\hat{p}_v} \right)\\
&\leq \frac{(2n - j_v)(\gammat+1)\epsilon}{8n(\gammat+2)}\max\left(\frac{\hat{q}_v'}{\hat{p}_v'} , \frac{\hat{q}_v}{\hat{p}_v}\right),
\end{align*}
where we use Lemma \ref{lem:p_p'_relationship} in the first inequality, as well as the fact that $\frac{\hat{q}_w}{\hat{p}_w} \leq \gammat \frac{\hat{q}_v}{\hat{p}_v}$ for any child $w$ of $v$. To obtain the second statement, note that since $\frac{\epsilon(\gammat+1)}{4(\gammat+2)} \leq \frac{\epsilon}{4} \leq \frac{1}{2}$, explicit calculation shows us that the first statement implies that
\[\left|\log\left(\frac{\hat{q}'_v}{\hat{p}'_v}\right)-\log\left(\frac{\hat{q}_v}{\hat{p}_v}\right)\right | \leq \log\left( 1+ 2\left(\frac{\epsilon(\gammat+1)}{4(\gammat+2)}\right)\right) = \log\left(1 + \frac{\epsilon(\gammat+1)}{2(\gammat+2)}\right).\]
\end{proof}

We are now ready to prove that sampling from $q'_v$ amount to sampling from $q_v$, up to a total variation distance of $\frac{\epsilon}{2}$. 
The proof of this lemma reduces to the previous bound on $|p_v - p_v'|$, and the technical result of Lemma \ref{lem:approximate_expectation} that can bound the difference of $\hat{q}_r'$ and $\hat{q}_r$. 
The constant $\delta$ was chosen precisely to obtain the following bound. 
\begin{lemma}\label{lem:q_approx}
The total variation distance of the distributions $q_v$ and $q'_v$ is bounded by $\frac{\epsilon}{2}$. 
\end{lemma}
\begin{proof}
The total variation distance of the two distributions may be written as
\begin{align*}
\text{TV}(q, q')&=\frac{1}{2}\sum_{v \in \mathcal{L}_{\mathcal{T}(r)}} |q_v - q_v'| \\&= \frac{1}{2}\sum_{v \in \mathcal{L}_{\mathcal{T}'(r)}} |q_v - q_v'| + \frac{1}{2}\sum_{v \in \mathcal{L}_{\mathcal{T}(r)} \backslash \mathcal{L}_{\mathcal{T'}(r)}} |q_v - q_v'|\\
&= \frac{1}{2}\sum_{v \in \mathcal{L}_{\mathcal{T}'(r)}} |q_v - q_v'| + \frac{1}{2}\sum_{v \in \mathcal{L}_{\mathcal{T}(r)} \backslash \mathcal{L}_{\mathcal{T'}(r)}} q_v \\
&\leq \sum_{v \in \mathcal{L}_{\mathcal{T}'(r)}} |q_v - q_v'|,
\end{align*}
where the second equality holds since $\hat{q}'_v=0$ for nodes $v \notin \mathcal{L}_{\mathcal{T}'}(r)$, and the inequality holds because 
\[\sum_{v \in \mathcal{L}_{\mathcal{T}(r)} \backslash \mathcal{L}_{\mathcal{T'}(r)}} q_v = 1 - \sum_{v \in \mathcal{L}_{\mathcal{T}'}(r) } q_v = \sum_{v \in \mathcal{L}_{\mathcal{T}'}(r) } q'_v - q_v \leq \sum_{v \in \mathcal{L}_{\mathcal{T}'}(r) } |q_v - q'_v|. \]
Now, it follows that, denoting the root of $\mathcal{T}'$ by $r$,
\begin{align*}
\text{TV}(q, q')&\leq\sum_{v \in \mathcal{L}_{\mathcal{T}'(r)}} |q_v - q_v'| 
\\ &= \sum_{v \in \mathcal{L}_{\mathcal{T}'(r)}} \left|\frac{\hat{q}_v}{\hat{q}_r}  - \frac{\hat{q}_v'}{\hat{q}_r'} \right|
\\&\leq \frac{1}{\hat{q}_r'}\sum_{v \in \mathcal{L}_{\mathcal{T}'(r)}} \left|\hat{q}_v  - \hat{q}_v' \right| + \sum_{v \in \mathcal{L}_{\mathcal{T}'(r)}} \left|\frac{\hat{q}_v}{\hat{q}_r}  - \frac{\hat{q}_v}{\hat{q}_r'}\right|
\\&= \frac{1}{\hat{q}_r'}\sum_{v \in \mathcal{L}_{\mathcal{T}'(r)}} \text{tr}(\sigma_v)|p_v-p_v'| + \sum_{v \in \mathcal{L}_{\mathcal{T}'(r)}} \frac{\hat{q}_v}{\hat{q}_r\hat{q}_r'}\left|\hat{q}_r'  - \hat{q}_r\right|
\\ &\leq \frac{1}{\hat{q}_r'}\sum_{v \in \mathcal{L}_{\mathcal{T}'}(r)} \text{tr}(\sigma_v)|p_v-p_v'|+ \sum_{v \in \mathcal{L}_{\mathcal{T}'(r)}} \frac{\hat{q}_v}{\hat{q}_r}\left(\frac{\epsilon(\gammat+1)}{2(\gammat+2)} \right) \\ &\leq \frac{\epsilon}{4(\gammat+2) \hat{q}_r'} \sum_{v \in \mathcal{L}_{\mathcal{T}'}(r)} p'_v \text{tr}(\sigma_v) + \frac{\epsilon(\gammat+1)}{2(\gammat+2) }\\&= \frac{\epsilon}{4(\gammat+2)\hat{q}_r'} \sum_{v \in \mathcal{L}_{\mathcal{T}'}(r)}\hat{q}'_v+ \frac{\epsilon(\gammat+1)}{2(\gammat+2) } \\&< \frac{\epsilon}{2}
\end{align*}
as desired.
Here, the second inequality comes from applying a direct triangle inequality, and the third inequality comes from Lemma \ref{lem:approximate_expectation}, applied to the root node $r$. 
In particular, since $\hat{p}_r= \hat{p}_r'=1$, we have that $\left |\log(\hat{q}'_r)- \log(\hat{q}_r) \right| \leq \log(1 + \frac{\epsilon(\gammat+1)}{2(\gammat+2)})$, which implies that $|\hat{q}_r - \hat{q}'_r| \leq \frac{\epsilon(\gammat+1)}{2(\gammat+2)}\hat{q}'_r$.
The last inequality comes from applying Lemma \ref{lem:p_p'_relationship} and the fact that $\sum_{v \in \mathcal{L}_{\mathcal{T}'(r)}} \frac{\hat{q}_v}{\hat{q}_r} \leq 1.$
We further note that $\sum_{v \in \mathcal{L}_{\mathcal{T}'}(r)}\hat{q}'_v = \hat{q}'_r$ to arrive at the final expression. 
\end{proof}

Having developed approximate analogs of the original distributions on $\mathcal{T}$, we can prove Theorem~\ref{thm:efficient_gibbs_sampling}. 
\\
\begin{proof}[Theorem~\ref{thm:efficient_gibbs_sampling}]
This proof will proceed in two steps. 
Firstly, we will show that all four conditions of Theorem~\ref{thm:simulate_distr} are satisfied for the distributions $p'$ and $q'$. 
Then, we will show that the ability to sample from $q'$ yields the desired Gibbs sampling algorithm. 

We will establish each condition of $p'$ and $q'$ necessary for Theorem \ref{thm:simulate_distr}. 
We also have that condition (2) is easily satisfied. 
Indeed, by definition, for any leaf node $v \in \mathcal{T}'$
\[\frac{\hat{q}_v'}{\hat{p}'_v} = \text{tr}(\sigma_v),\]
which can be exactly calculated given $v$. 
Condition (3) is also satisfied. 
Indeed, for any internal node $v$ and a child $v_k$, $\frac{p'_{v_k}}{p'_v}$ is the conditional probability of reaching $v_k$ by repeated application of Algorithm \ref{alg:pinning_subroutine_rounded} given that the algorithm reaches $v$. 
A single application of Algorithm 22 therefore suffices to sample the children of $v$ according to this probability distribution. 
Now notice that both conditions (1) and (4) follow from the statement of Lemma \ref{lem:approximate_expectation}, which implies that
\[\left|\log\left(\frac{\hat{q}'_v}{\hat{p}'_v}\right)-\log\left(\frac{\hat{q}_v}{\hat{p}_v}\right)\right | \leq \log\left( 1+ \frac{\epsilon}{2}\right).\]
As we have already demonstrated an efficient algorithm to calculate $\frac{\hat{q}_v}{\hat{p}_v}$ up to constant multiplicative error, this statement shows that the same algorithm calculates $\frac{\hat{q}_v'}{\hat{p}_v'}$ up to constant multiplicative error. Condition (1) is therefore satisfied. 
Meanwhile, condition (4) can be obtained from the triangle inequality, since for a node $v$ and a child $w$, 
\begin{equation*}
\begin{aligned}
\left|\log\left(\frac{\hat{q}'_v}{\hat{p}'_v}\right)-\log\left(\frac{\hat{q}'_w}{\hat{p}'_w}\right)\right |
&\leq \left|\log\left(\frac{\hat{q}'_v}{\hat{p}'_v}\right)-\log\left(\frac{\hat{q}_v}{\hat{p}_v}\right)\right | + \left|\log\left(\frac{\hat{q}_v}{\hat{p}_v}\right)-\log\left(\frac{\hat{q}_w}{\hat{p}_w}\right)\right | \\
& + \left|\log\left(\frac{\hat{q}'_w}{\hat{p}'_w}\right)-\log\left(\frac{\hat{q}_w}{\hat{p}_w}\right)\right | \\
&\leq 2 \log \left(1 + \frac{ \epsilon}{2}\right) + \log (\gammat), 
\end{aligned}
\end{equation*}
where the middle term is at most $\log(\gammat)$ since condition (4) is satisfied for $\hat{p}$ and $\hat{q}$. 
We may therefore apply Theorem \ref{thm:simulate_distr} and obtain a $\text{poly}(n, \log(\epsilon^{-1}), \log(k))$ algorithm to sample a distribution $q''$ that is $q'$ up to $\frac{\epsilon}{2}$ total variation distance.
The number of children of each node was bounded in Lemma \ref{lem:tree_bound} so that $\log(k)$ is polynomial in $n$.
By Lemma \ref{lem:q_approx}, the total variation distance between $q''$ and $q$ is at most $\epsilon$.
It follows that
\begin{align*}
\left\lVert \sum_{v \in \mathcal{L}_\mathcal{T}(r)} q''_v \frac{\sigma_v}{\text{tr}(\sigma_v)}- \rho_\beta\right\rVert_{\text{tr}} &= \frac{1}{2}\left\lVert \sum_{v \in \mathcal{L}_\mathcal{T}(r)} q''_v \frac{\sigma_v}{\text{tr}(\sigma_v)}- \rho_\beta\right\rVert_1 \\
&= \frac{1}{2}\left\lVert \sum_{v \in \mathcal{L}_\mathcal{T}(r)} q''_v \frac{\sigma_v}{\text{tr}(\sigma_v)}- \sum_{v \in \mathcal{L}_\mathcal{T}(r)} q_v \frac{\sigma_v}{\text{tr}(\sigma_v)}\right\rVert_1 \\&\leq \frac{1}{2} \sum_{v \in \mathcal{L}_\mathcal{T}(r)}| q''_v-q_v|\\
&\leq \epsilon,
\end{align*}
giving the desired trace distance bound. 

\end{proof}

\section{High-Temperature SYK is Non-Gaussian}\label{sec:SYK_non-gaussian}

In this section, we present an example of fermionic Hamiltonians whose Gibbs states at high temperature (small $\beta$) cannot be expressed as a probabilistic mixture of fermionic Gaussian states. Specifically, we focus on the Sachdev–Ye–Kitaev (SYK) model~\cite{sachdev1993gapless,kitaev2015simple}. The $q$-local SYK model for even $q$ is defined as
\begin{equation}
H =  i^{q/2} \sum_{1\leq k_1<...<k_q\leq 2n} J_{k_1 \dots k_q}\gamma_{k_1}\dots \gamma_{k_q}. 
\end{equation}
where $\{J_{k_1 \dots k_q}\}_{k_1 \dots k_q}$ are independent Gaussian variables with mean 0 and variance 1. We use $H\sim \text{SYK}_q$ to denote sample a random $\text{SYK}_q$ Hamiltonian according to the randomness of $\{J_{k_1 \dots k_q}\}_{k_1 \dots k_q}$.

It was shown that Gaussian states cannot achieve constant approximation ratio for the \textbf{ground state} to the SYK model~\cite{hastings2022optimizing,herasymenko2023optimizing,ding2025optimizing}. In this section, we show that those results can be generalized to Gibbs states of the SYK model at constant temperature, formally proven in Theorem \ref{thm:SYK_gaussian}.

% \begin{theorem}[Informal version of Theorem \ref{thm:SYK_gaussian}]
    
% Let $q \geq 12$ be an even integer.
% For any constant temperature $\beta$, there is high probability over $H \sim \text{SYK}_q$ that there exists a $\beta' < \beta$ for which 
% \begin{equation}
% \rho_{\beta'} = \frac{e^{-\beta' H}}{\text{tr}\left(e^{-\beta' H}\right)}
% \end{equation}
% is not a Gaussian state.
% \end{theorem}

In particular, $\text{SYK}_q$ consists of constant locality Hamiltonians, but not of constant degree, so our structural result (Theorem \ref{thm:structural}) 
does not apply. 
Theorem \ref{thm:SYK_gaussian}  shows that a large class of random Hamiltonians of constant locality do not exhibit death of entanglement.

Here, we use $negl(n)$ to denote negligible functions  which decay exponentially fast with $n$ 
and $||\cdot ||$ to denote operator norms.
Furthermore, we use the asymptotic notation $\omega$ so that $f = \omega(g)$ means that $g = o(f)$, i.e. $\frac{g}{f}$ asymptotically vanishes.

\begin{lemma}\label{lem:SYK_operator_norm_upper_bound}
There exists a constant $c > 0$ such that with probability 1-$\text{negl}(n)$,
\begin{equation}
\lVert H \rVert \leq cn^{(q+1)/2},
\end{equation}
for some sufficiently large constant $c > 0.$
In other words, with all but negligible probability over $H \sim \text{SYK}_q$, $\lVert H\rVert = O(n^{(q+1)/2})$. 
\end{lemma}
\begin{proof}
We apply the Gaussian matrix Chernoff bound to
\begin{equation}
H = i^{q/2}\sum_{1 \leq k_1 < \dots < k_q \leq 2n} J_{k_1 \dots k_q}\gamma_{k_1}\dots \gamma_{k_q}, 
\end{equation}
which applies since $J_{k_1\dots k_q}$ is a standard normal random variable and $i^{q/2}\gamma_{i_1}\dots\gamma_{i_q}$ is Hermitian.
Applying the bound directly, we obtain
\begin{align}
\mathbb {P}\left[\lVert H\rVert> t\right]&\leq 2^n \exp\left(-\frac{t^2}{2\sigma^2}\right),
\end{align}
where 
\begin{equation}
\sigma^2 = \left\|\sum_{1 \leq k_1< \dots< k_q\leq 2n}(i^{q/2}\gamma_{k_1}\dots \gamma_{k_q})^2\right\|_2 = {2n \choose q}
\end{equation}
where $\|\cdot\|_2$ is the matrix $2$-norm.

Since ${2n \choose q}\leq  Dn^q$ for sufficiently large constant $D>0$. 
Choosing $t = cn^{(q+1)/2}$ for sufficiently large constant $c >0$ such that $c^2 > 2\ln(2)D$, we find that $\mathbb {P}\left[\lVert H\rVert> t\right] \leq 2^n\exp\left(-\frac{c^2}{2D}n\right) = \text{negl}(n)$. 

\end{proof}

\begin{lemma}\label{lem:trace_syk}
For any $H \sim \text{SYK}_q$, $\text{tr}(H) = 0$.
\end{lemma}
\begin{proof}
For any Majorana string $\gamma_{i_1}\dots \gamma_{i_k}$ which is not equal to a scalar, $\text{tr}(\gamma_{i_1}\dots \gamma_{i_k}) = 0$. Since the $\text{SYK}_q$ distribution consists of Hamiltonians that are sums of nontrivial Majorana strings, we may conclude that $\text{tr}(H) = 0$. 
\end{proof}

\begin{lemma}\label{lem:trace_squared_syk}
There exists constants $0<c_1<c_2$ such that with probability $1-\text{negl}(n)$, 
\begin{equation}
c_1n^q \leq \frac{\text{tr}(H^2)}{2^n}\leq  c_2n^q.
\end{equation}
In other words, with all but negligible probability over $H \sim \text{SYK}_q$, $\frac{\text{tr}(H^2)}{2^n} = \Theta(n^q)$. 
\end{lemma}
\begin{proof}
The Hamiltonian $H$ is of the form
\begin{equation}
H = i^{q/2}\sum_{1 \leq k_1 < \dots < k_q \leq 2n} J_{k_1 \dots k_q}\gamma_{k_1}\dots \gamma_{k_q}. 
\end{equation}
Aftering squaring the Hamiltonian and expanding, the scalar term of $H^2$ is 
\begin{equation}
\sum_{1 \leq k_1 < \dots < k_q \leq 2n} J_{k_1 \dots k_q}^2 I.
\end{equation}
whose trace equals to  $\text{tr}(H^2)$, since
 the trace of non-trivial Majorona string is 0.    
Since the trace of $I$ is $\text{tr}(I) = 2^n$, we conclude that 
\begin{equation}
\frac{\text{tr}(H^2)}{2^n} = \sum_{1 \leq k_1 < \dots < k_q \leq 2n} J_{k_1 \dots k_q}^2.
\end{equation}
Note that the value  $\sum_{1 \leq k_1 < \dots < k_q \leq 2n} J^2_{k_1 \dots k_q}$ follows  the  Chi-Square distribution thus exhibits strong concentration around its mean (the Laurent-Massart bound).
Thus with probability $1 - \text{negl}(n)$, this sum is $\Omega(n^q)$. 
\end{proof}

Finally, we mention a result that bounds the maximum energy that any Gaussian state can attain with high probability over the SYK distribution. 
\begin{lemma}[Lemma 9, \cite{herasymenko2023optimizing}]\label{lem:gaussian_energy_bound}~
There exists a constant $C>0$ for which with probability $1-\text{negl}(n)$,
\begin{equation}
\left|\text{tr}(\tau H)\right| \leq Cn^{q/4+1}.
\end{equation}
 In other words, with all but negligible probability over $H \sim \text{SYK}_q$ and for all Gaussian states $\tau$, $\left|\text{tr}(\tau H)\right| = O(n^{q/4+1})$.
\end{lemma}
% \begin{proof}
% For any fixed state $\ket{\psi}$, consider the expression
% \begin{equation}
% \braket{\psi|  H |\psi}.
% \end{equation}
% For the SYK model, this must be a Gaussian with mean \jiaqing{what do you mean?}, since it is a homogeneous linear expression in standard Gaussian random variables $J_{i_1\dots i_q}$.
% Furthermore, \color{red} add citation based on commutation index (eg sparse SYK optimization, matthew ding) \color{black}  it has variance $O(n^{q/2})$. 
% For a fixed $H$, the probability that it reaches an energy of $t$ is 
% \begin{equation}
% \mathbb{P}\left[\bra{\psi|  H |\psi} > t\right] =\exp\left(-\Omega\left(\frac{t^2}{n^{q/2}}\right)\right). 
% \end{equation}
% Meanwhile, the number of gates required to build an $\epsilon$-net for pure Gaussian states is $\exp(\Theta(n^2 \log(n)))$, for any $\epsilon = \frac{1}{\text{poly}(n)}$. 
% Hence, if the energy, asymptotically, $t$ can be attained by a Gaussian state, then $t - \frac{1}{\text{poly(n)}} $ may also be obtained by one of these states. 
% Hence, it must be true that 
% \begin{equation}
% \frac{t^2}{n^{q/2}}  = O(n^2\log(n)),
% \end{equation}
% or by a union bound over the states in the $\epsilon$-net $S$ there would be a negligible chance that 
% \begin{align}
% \sum_{\psi \in S} \mathbb{P}\left[\bra{\psi|  H |\psi} > t\right] &= \exp\left(-\Omega\left(\frac{t^2}{n^{q/2}}\right) + \Omega(n^2\log(n))\right)\\ &= \text{negl}(n). 
% \end{align}
% We conclude Gaussian states obtain an energy for which $t = \omega(n^{q/4+1}\sqrt{\log(n)})$ with negligible probability. 
% This completes the proof
% \end{proof}

\begin{theorem}\label{thm:SYK_gaussian}
For any constant $\epsilon>0$ and inverse temperature $\beta= \Omega(n^{-3q/4 + 1+\epsilon}),
$ with high probability 
over $H \sim \text{SYK}_q$  there exists some $\beta_0 < \beta$ for which
\begin{equation}
\rho_{\beta_0} = \frac{e^{-\beta_0 H}}{\text{tr}\left(e^{-\beta_0 H}\right)}.
\end{equation}
is not a convex combination of Gaussian states. 
\end{theorem}
\begin{proof}
We begin with the notation
\begin{equation}
\langle -H \rangle_\beta = \text{tr}(-\rho_\beta H). 
\end{equation}
By explicit calculation, we see that 
\begin{equation}
\frac{d}{d \beta} \langle -H \rangle_\beta = \langle H^2 \rangle_\beta-\langle H \rangle_\beta^2 
\end{equation}
and
\begin{equation}
\frac{d^2}{d^2 \beta} \langle -H \rangle_\beta = -\langle H^3 \rangle_\beta + 3\langle H \rangle_\beta \langle H^2 \rangle_\beta  - 2\langle H \rangle_\beta^3. \label{eq:73}
\end{equation}
Note that 
\begin{align}
D &\coloneqq \left(\frac{d}{d\beta}\langle -H \rangle_\beta\right)\Big\vert_0 \\
&= \frac{\text{tr}(H^2)}{2^n} - \left(\frac{\text{tr}(H)}{2^n}\right)^2\\
&= \Theta(n^q),\label{eq:76}
\end{align}
where the last equality holds by Lemma \ref{lem:trace_syk} and Lemma \ref{lem:trace_squared_syk} with all but negligible probability.
Meanwhile, we may upper bound the second derivative at any $\beta$ with
\begin{align}
\left|\frac{d^2}{d^2 \beta} \langle -H \rangle_\beta\right| &=O( \lVert H \rVert^3)\\
&= O(n^{3(q+1)/2})
\end{align}
since in general $\langle H \rangle_\beta  = \text{tr}(\rho_\beta H)\leq \lVert H \rVert$. 
The second inequality holds by Lemma \ref{lem:SYK_operator_norm_upper_bound} with all but negligible probability over $\text{SYK}_q$.

Let $S = \Theta(\lVert H \rVert^3)$ be the upper bound on the second derivative, and let $0 <\beta_0 \leq \frac{D}{S}$.
Note that for any $0\leq \beta \leq \beta_0$, according to Eq.~(\ref{eq:73}) and Eq.~(\ref{eq:76}),  we have that 
\begin{align}
    \frac{d\langle -H \rangle_{\beta}}{d\beta} \geq D-\beta S.
\end{align}
Then
\begin{align}
\langle -H \rangle_{\beta_0}&\geq \int_{0}^{\beta_0} (D - \beta S)d\beta\\
&= D\beta_0 - \frac{S\beta_0^2}{2}\\
&\geq D\beta_0 - \frac{D\beta_0}{2}\\
&= \frac{D\beta_0}{2}. 
\end{align}
where the first inequality holds because $D - \beta S$ is a lower bound on the derivative of $\langle -H \rangle_{\beta}$ at $\beta$, and $\langle -H \rangle_{\beta_0} \left|_{\beta_0=0}\right. =tr(H)=0$. 
We conclude that with all but negligible probability over $H \sim \text{SYK}_q$
\begin{align}
\text{tr}(\rho_{\beta_0} (-H)) &=\Omega\left(D\beta_0\right) = \Omega\left(n^q \beta_0\right).
\end{align}
The maximum value given for $\beta_0$ is 
\begin{align}
D/S &= \Omega(n^{q - 3(q+1)/2}) 
\\&= \Omega(n^{-q/2 - 3/2})
\\&= \omega(n^{-3q/4 + 1}),
\end{align}
where the last bound follows when $q \geq 12$. 
Hence, we may at least choose $\beta_0 = \omega(n^{-3q/4 + 1})$. 
We now show that for this $\beta_0$, $\rho_{\beta_0}$ achieves an energy not asymptotically achievable by any Gaussian state, so as a consequence $\rho_{\beta_0}$ cannot be a convex combination of Gaussian states. 
By Lemma \ref{lem:gaussian_energy_bound}, the best energy a Gaussian state $\tau$ can achieve is
\begin{equation}
|\text{tr}(\tau H)| = O(n^{q/4 + 1}). 
\end{equation}
Since $D/S = \omega(n^{-3q/4 + 1})$, we may choose $\beta_0= \omega(n^{-3q/4 + 1})$, obtaining
\begin{align}
\text{tr}(\rho_{\beta_0}(-H)) &= \Omega(n^q\beta_0)\\
&= \omega(n^{q/4+1})\\
&= \omega(|\text{tr}(\tau H)|)\\
&= \omega(\text{tr}(\tau (-H)))
\end{align}

We conclude that with all but negligible probability over $H \sim \text{SYK}_q$, for any constant $\epsilon>0$ and $\beta = \Omega(n^{-3q/4 + 1 + \epsilon})$, $\rho_{\beta_0}$ is not a distribution of Gaussian states. 
Since $\beta_0$ can be chosen to be asymptotically smaller than $\beta$ this proves the result. 

\end{proof}

\begin{corollary}
For typical normalizations of the SYK model, such as $\frac{1}{n^{q/2}}H$ and $\frac{1}{n^{(q-1)/2}} H$ for $H \sim \text{SYK}_q$, there is no constant $\beta$ such that for any $\beta' < \beta$, the Gibbs state $\rho_{\beta'}$ is a convex combination of Gaussian states. 
\end{corollary}
\begin{proof}
This statement follows from Theorem \ref{thm:SYK_gaussian}.
In the theorem, for any $\beta = \omega(n^{-3q/4+1+\epsilon})$, we have that
\begin{equation}
\frac{e^{-\beta H}}{\text{tr}(e^{-\beta H})}
\end{equation}
is not a convex combination of Gaussian states. 

Dividing $H$ by the normalization factor $C = n^{q/2}$ or $C = n^{(q-1)/2}$, the same result holds when $\beta$ is scaled up by the same normalization factor. 
The bound on $\beta$ in the condition of the Theorem \ref{thm:SYK_gaussian} becomes $\Omega(n^{-q/4+1+\epsilon})$ or $\Omega(n^{-q/4 + 1/2 + \epsilon})$ in either case.
For $q \geq 12$, and choosing $\epsilon = 1$, any constant $\beta$ satisfies this bound, and as a consequence there is some $\beta_0 < \beta$ for which
\begin{equation}
\frac{e^{-\beta_0 (H/C)}}{\text{tr}(e^{-\beta_0 (H/C)})}
\end{equation}
is not a distribution of Gaussian states. 
\end{proof}

\paragraph*{Acknowledgments}
We thank Ewin Tang, Ainesh Bakshi,  Robbie King and Eric Anschuetz for helpful discussions.
Y.C. is supported by the Mellon Mays Undergraduate Fellowship.
J.J. is supported by NSF CCF-2321079, MURI Grant FA9550-18-1-0161, NSF CAREER award CCF-2048204, and the IQIM, an  NSF Physics Frontiers Center (NSF Grant PHY-1125565). 

\newpage
\bibliographystyle{unsrt}
\bibliography{ref}
\end{document}